\documentclass[runningheads]{llncs}

\usepackage{booktabs}   
\usepackage{subcaption} 
\usepackage{cmap}

\usepackage{wrapfig}

\usepackage{shuffle}
\usepackage{syntax}

\usepackage{url}
\usepackage{graphicx}
\usepackage{amsmath} 
\usepackage{amssymb} 
\usepackage{amsfonts} 
\usepackage{calligra}
\usepackage{enumitem}

\usepackage{etoolbox}
\usepackage{mathpartir}
\usepackage{mathtools}
\usepackage{refcount}
\usepackage{stmaryrd}
\usepackage{enumitem}
\usepackage{listings}
\usepackage{caption}
\usepackage{flushend}

\usepackage{xcolor}

\usepackage{algorithm}
\usepackage[noend]{algpseudocode}

\definecolor{bluekeywords}{rgb}{0,0,1}
\definecolor{greencomments}{rgb}{0,0.5,0}
\definecolor{redstrings}{rgb}{0.64,0.08,0.08}
\definecolor{xmlcomments}{rgb}{0.5,0.5,0.5}
\definecolor{types}{rgb}{0.17,0.57,0.68}

\usepackage{listings}
\usepackage{array}
\newcolumntype{L}[1]{>{\raggedright\let\newline\\\arraybackslash\hspace{0pt}}m{#1}}
\newcolumntype{C}[1]{>{\centering\let\newline\\\arraybackslash\hspace{0pt}}m{#1}}
\newcolumntype{R}[1]{>{\raggedleft\let\newline\\\arraybackslash\hspace{0pt}}m{#1}}

\usepackage{pgf}
\usepackage{tikz}
\usetikzlibrary{tikzmark}
\usetikzlibrary{calc,arrows,automata,fit}
\usetikzlibrary{positioning,shapes,shadows}
\usetikzlibrary{decorations.pathmorphing,snakes}

\usepackage{comment}
\usepackage{mathdots}

\usepackage{algpseudocode}
\makeatletter
\renewcommand{\ALG@beginalgorithmic}{\footnotesize}
\makeatother
\usepackage{eqnarray}
\usepackage{alltt}
\usepackage{textcomp}
\usepackage{wrapfig,lipsum}
\usepackage{mathsprograms}

\usepackage{listings}
\lstdefinelanguage{JavaScript}{%
  keywords = { async, await, break, case, catch, class, const, continue, debugger, default, delete, do, each, 
  else, export, finally, for, function, if, import, in, instanceof, let, new, of, return, switch, this, throw, 
  try, typeof, var, void, while, with, yield },
  morecomment = [l]{//},
  morecomment = [s]{/*}{*/},
  morestring  = [b]',
  morestring  = [b]",
  sensitive   = true,
}

\lstdefinelanguage{Java10}{
  language      = Java,
  morekeywords  ={ var, async },
}

\begin{document}

\title{Automated Synthesis of Asynchronizations\thanks{This work is supported in part by the European Research Council (ERC) under the Horizon 2020 research and innovation programme (grant agreement No 678177).}}
%
%
\author{Sidi Mohamed Beillahi \inst{1} \and
Ahmed Bouajjani \inst{2} \and
Constantin Enea \inst{3} \and Shuvendu Lahiri \inst{4}}
\authorrunning{S.M. Beillahi, A. Bouajjani, C. Enea, and S. Lahiri.}
\institute{University of Toronto, Canada \\
\email{sm.beillahi@utoronto.ca} \and 
Universit\'e Paris Cit\'e, IRIF, CNRS, Paris, France \\
\email{abou@irif.fr} \and 
LIX, Ecole Polytechnique, CNRS and Institut Polytechnique de Paris, France \\
\email{cenea@irif.fr} \and 
Microsoft Research Lab - Redmond \\
\email{shuvendu@microsoft.com}}
%
\maketitle

    \begin{abstract}
        \vspace{-10pt}
       Asynchronous programming is widely adopted for building responsive and efficient software, and 
       modern languages such as C\# provide async/await primitives to simplify the use of asynchrony.
       In this paper, we propose an approach for refactoring a sequential program into an asynchronous program that uses async/await, called asynchronization.
       The refactoring process is parametrized by a set of methods to replace with asynchronous versions, and it is constrained
       to avoid introducing data races. 
       We investigate the delay complexity of enumerating all data race free asynchronizations, which quantifies the delay
       between outputting two consecutive solutions. We show that this is polynomial time modulo an oracle for solving reachability 
       in sequential programs. We also describe a pragmatic approach based on an interprocedural data-flow analysis with polynomial-time 
       delay complexity. The latter approach has been implemented and evaluated on a number of non-trivial C\# programs extracted from
       open-source repositories.
       \vspace{-5pt}
    \end{abstract}

    \everymath{\displaystyle}

\vspace{-15pt}
\section{Introduction}
\vspace{-5pt}

Asynchronous programming is widely adopted for building responsive and efficient software. 
%
As an alternative to 
explicitly registering callbacks with asynchronous calls, C\# 5.0~\cite{DBLP:conf/ecoop/BiermanRMMT12} introduced the \texttt{async}/\texttt{await} primitives. These primitives allow the programmer to write code in a familiar sequential style without explicit callbacks. An asynchronous procedure, marked with \texttt{async}, returns a task object that the caller uses to ``await'' it. 
Awaiting may suspend the execution of the caller, but does not block the thread it is running on. The code after \texttt{await} is the continuation called back when the callee result is ready. 
This paradigm has become popular across many languages, C++, JavaScript, Python.

The \texttt{async}/\texttt{await} primitives introduce concurrency which is notoriously complex. The code in between a call and a matching \texttt{await} (referring to the same task) may execute before some part of the awaited task 
or after the awaited task finished. For instance, on the middle of Fig.~\ref{fig:example}, the assignment \texttt{y=1} at line~\ref{ln:introy} 
 can execute before or after \texttt{RdFile} finishes. The \texttt{await} for \texttt{ReadToEndAsync} in \texttt{RdFile} (line~\ref{ln:AwaitT6}) may suspend \texttt{RdFile}'s execution because \texttt{ReadToEndAsync} did not finish, and pass the control to \texttt{Main} which executes \texttt{y=1}. 
If \texttt{ReadToEndAsync} finishes before this \texttt{await} executes, then the latter has no effect and \texttt{y=1} gets executed after \texttt{RdFile} finishes.
The resemblance with sequential code can be especially deceitful since this non-determinism is opaque. It is common that \texttt{await}s are placed immediately after the corresponding call which limits the benefits that can be obtained from executing steps in the caller and callee concurrently~\cite{DBLP:conf/icse/OkurHDD14}. 

In this paper, we address the problem of writing efficient asynchronous code that uses \texttt{async}/\texttt{await}. We propose a procedure for automated synthesis of asynchronous programs \emph{equivalent} to a given synchronous (sequential) program $\aprog$. This can be seen as a way of refactoring synchronous code to asynchronous code. Solving this problem in its full generality would require checking equivalence between arbitrary programs, which is known to be hard. Therefore, we consider a restricted space of asynchronous program candidates defined by substituting synchronous methods in $\aprog$ with asynchronous versions (assumed to be behaviorally equivalent). The substituted methods are assumed to be leaves of the call-tree (they do not call any method in $\aprog$). Such programs are called \emph{asynchronizations} of $\aprog$. A practical instantiation is replacing IO synchronous calls for reading/writing files or managing http connections with asynchronous versions.


\begin{figure}[t]
    \lstset{basicstyle=\ttfamily\scriptsize,numbers=left,
        stepnumber=1,numberblanklines=false,mathescape=true,morekeywords  ={ async, Task }}
    \hspace{1mm}
\begin{minipage}[c]{0.3\textwidth}
\begin{lstlisting}[numbersep=2pt]
void Main(string f) {
 x = 0;
 int val = RdFile(f);
 y = 1; 

 int r = x;
 Debug.Assert(r == val); $\label{ln:assertSeq}$}
 
int RdFile(string f) {
 var rd=new StreamReader(f);
 string s = rd.ReadToEnd();  
 int r1 = x;

 x = r1 + s.Length;
 return s.Length;    }
\end{lstlisting}
\end{minipage}
\hfill
\begin{minipage}[c]{0.42\textwidth}
\begin{lstlisting}[xleftmargin=2.5mm,numbersep=2pt]
async Task Main(string f) {
 x = 0; $\label{ln:WriteX0}$
 Task<int> t1 = RdFile(f); $\label{ln:InvocT2}$
 y = 1; $\label{ln:introy}$
 int val = await t1;  $\label{ln:AwaitT3}$          
 int r = x;   $\label{ln:ReadX1}$
 Debug.Assert(r == val);  }  
 
async Task<int> RdFile(string f) {
 var rd = new StreamReader(f);
 Task<string> t=rd.ReadToEndAsync(); $\label{ln:InvocT6}$
 int r1 = x;             $\label{ln:WriteX}$
 string s = await t;     $\label{ln:AwaitT6}$
 x = r1 + s.Length; $\label{ln:return:InvocT6}$
 return s.Length;                }
\end{lstlisting}
\end{minipage}
\begin{minipage}[c]{0.23\textwidth}
    \includegraphics[width=1\linewidth]{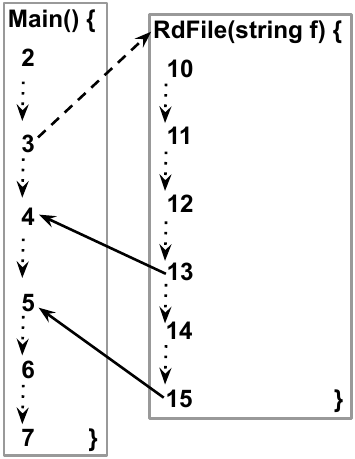}
\end{minipage}
\vspace{-0.3cm}
\caption{Synchronous and asynchronous C\# programs ({\tt x}, {\tt y} are static variables).} 
\label{fig:example}
\vspace{-.5cm}
\end{figure}

For instance, the sequential C\# program on the left of Fig.~\ref{fig:example} contains a \texttt{Main} that invokes a method \texttt{RdFile} that returns the length of the text in a file. The file name input to \texttt{RdFile} is an input to \texttt{Main}. The program uses a variable \texttt{x} to aggregate the lengths of all files accessed by \texttt{RdFile}; this would be more useful when \texttt{Main} calls \texttt{RdFile} multiple times which we omit for simplicity. Note that this program passes the assertion at line~\ref{ln:assertSeq}.
The time consuming method \texttt{Read\-ToEnd} for reading a file
is an obvious choice for being replaced with an equivalent \emph{asynchronous} version whose name is suffixed with \texttt{Async}.
Performing such tasks asynchronously can lead to significant performance boosts. 
The program on the middle of Fig.~\ref{fig:example} is an example of an asynchronization defined by this substitution.
 The syntax of \texttt{async}/\texttt{await} imposes that every method that transitively calls one of the substituted methods, i.e., \texttt{Main} and \texttt{RdFile}, must also be declared as asynchronous. Then, every asynchronous call must be followed by an \texttt{await} that specifies the control location where that task should have completed. 
 For instance, the \texttt{await} for \texttt{Read\-ToEndAsync} is placed at line~\ref{ln:AwaitT6} since the next instruction (at line~\ref{ln:return:InvocT6}) uses the computed value. Therefore, synthesizing such refactoring reduces to finding a correct placement of \texttt{await}s (that implies equivalence) for every call of a method that transitively calls a substituted method (we do not consider ``deeper'' refactoring like rewriting conditionals or loops). 


We consider an equivalence relation between a synchronous program and an asynchronization that corresponds to absence of data races in the asynchronization. Data race free asynchronizations are called \emph{sound}. Relying on absence of data races avoids reasoning about equality of sets of reachable states which is harder in general, and an established compromise in reasoning about concurrency.  
For instance, the asynchronization  
in Fig.~\ref{fig:example} is sound because the call to \texttt{RdFile} accessing \texttt{x} finishes before the read of \texttt{x} in \texttt{Main} (line~\ref{ln:ReadX1}). Therefore, accesses to \texttt{x} are performed in the same order as in the synchronous program.

The asynchronization on 
the right of Fig.~\ref{fig:example}  is not the only sound (data-race free) asynchronization of the program on the left. The \texttt{await} at line~\ref{ln:AwaitT6} can be moved one statement up (before the read of \texttt{x}) and the resulting program remains equivalent to the sequential one. In this paper, we investigate the problem of enumerating \emph{all} sound asynchronizations of a sequential program $P$ w.r.t. substituting a set of methods with asynchronous versions. This makes it possible to deal separately with the problem of choosing the best asynchronization in terms of performance based on some metric (e.g., performance tests). 

\begin{figure}[t]
\lstset{basicstyle=\ttfamily\scriptsize,numbers=left,
        stepnumber=1,numberblanklines=false,mathescape=true,morekeywords  ={ async, Task }}
\hspace{3mm}
\begin{minipage}[t]{0.30\textwidth}
\begin{lstlisting}[numbersep=2pt]
async Task Main() {
 var t1 = Foo();

 await t1;
                  }
                  
async Task Foo() {
 var t = IO();

 Thread.Sleep(200);

 Thread.Sleep(200); 
 await t;        }
 
async Task IO() {
 var t0 = Task.Delay(300);

 await t0;      }
\end{lstlisting}
\end{minipage}
\hfill
\begin{minipage}[t]{0.31\textwidth}
    \begin{lstlisting}[numbersep=2pt]
async Task Main() {
 var t1 = Foo();
 var t2 = IO();
 await t1;
 await t2;        }
 
async Task Foo() {
 var t = IO();
 
 Thread.Sleep(200);
 await t;
 Thread.Sleep(200); 
                  }
                  
async Task IO() {
 var t0 = Task.Delay(300);
    
 await t0;      }
\end{lstlisting}
\end{minipage}
\hfill
\begin{minipage}[t]{0.31\textwidth}
\begin{lstlisting}[xleftmargin=2.5mm,numbersep=2pt]
async Task Main() {
 var t1 = Foo();
 var t2 = IO();
 await t1;
 await t2;        }
 
async Task Foo() {
 var t = IO();
 await t;
 Thread.Sleep(200);
 
 Thread.Sleep(200); 
                 }
                 
async Task IO() {
 var t0 = Task.Delay(300);
 Thread.Sleep(150);
 await t0;      }
\end{lstlisting}
\end{minipage}

\vspace{-4mm}
    \begin{center}
    \includegraphics[width=\linewidth]{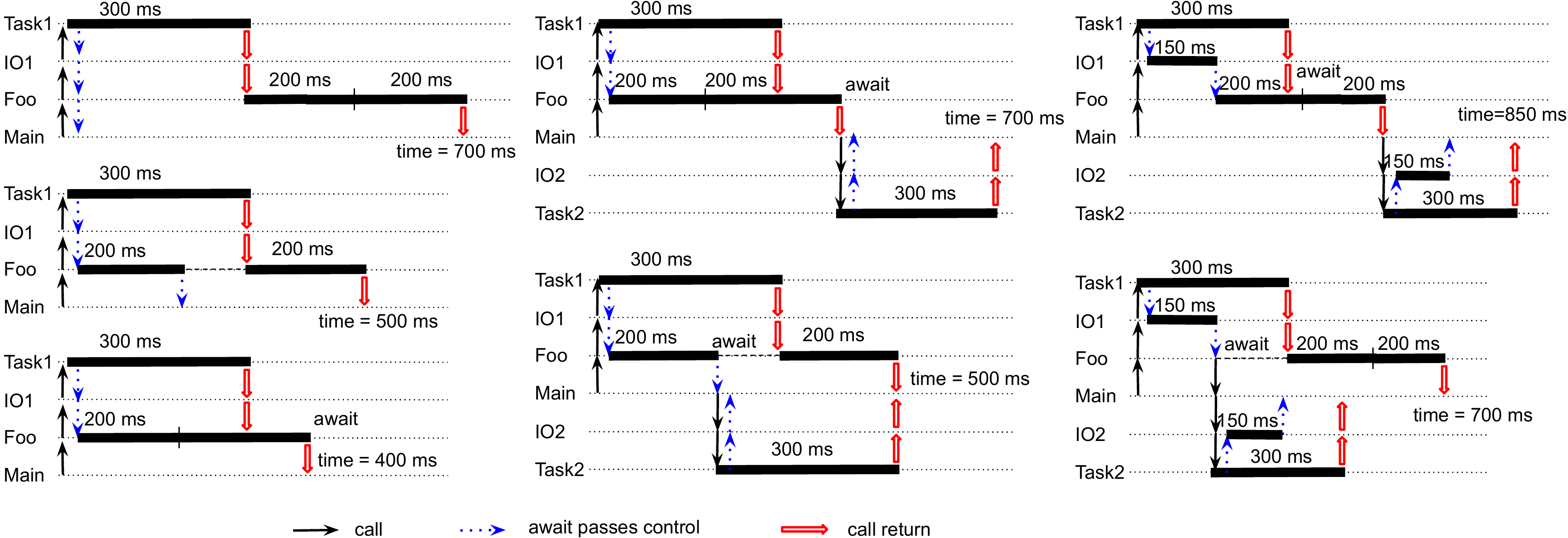}
    \end{center}
\vspace{-0.4cm}
\caption{Asynchronous C\# programs and executions. On the bottom, time durations of executing code blocks from the same method are aligned horizontally, and time goes from left to right. Vertical single-line arrows represent method call steps, dashed arrows represent \texttt{await}s passing control to the caller, and double-line arrows represent a call return. Total execution time is marked \texttt{time=...}.}
\label{fig:performance}
\vspace{-0.6cm}
\end{figure}


Identifying the most efficient asynchronization is difficult and can not be done syntactically. 
It is tempting to consider that increasing the distance between calls and matching \texttt{await}s so that more of the caller code is executed while waiting for an asynchronous task to finish increases performance. However, this is not true in general.
We use the programs in Fig.~\ref{fig:performance} to show that the best \texttt{await} placement w.r.t. performance depends on execution times of code blocks in between calls and \texttt{await}s in a non-trivial manner. Note that estimating these execution times, especially for IO operations like http connections, can not be done statically.

The programs in Fig.~\ref{fig:performance} use \texttt{Thread.Sleep(n)} to abstract sequential code executing in $n$ milliseconds and \texttt{Task.Delay(n)} to abstract an asynchronous call executing in $n$ milliseconds on a different thread. The functions named \texttt{Foo} differ only in the position of \texttt{await t}. We show that modifying this position worsens execution time in each case.
For the left program, best performance corresponds to maximal distance between \texttt{await t} in \texttt{Foo} and the corresponding call. This allows the \texttt{IO} call to execute in parallel with the caller, as depicted on the bottom-left of Fig.~\ref{fig:performance}. The executions corresponding to the other two positions of \texttt{await t} are given just above.
For the middle program, placing \texttt{await t} in between the two code blocks in \texttt{Foo} optimizes performance (note the extra \texttt{IO} call in \texttt{Main}): 
the \texttt{IO} call in \texttt{Foo} executes in parallel with the first code block in \texttt{Foo} and the \texttt{IO} call in \texttt{Main} executes in parallel with the second one. This is depicted on the bottom-middle of Fig.~\ref{fig:performance}. The execution above shows that placing \texttt{await t} as on the left (after the two code blocks) leads to worse execution time (placing \texttt{await t} immediately after the call is also worse). Finally, for the right program, placing \texttt{await t} immediately after the call is best (note that \texttt{IO} executes another code block before \texttt{await}). 
The \texttt{IO} call in \texttt{Main} executes in parallel with \texttt{Foo} as shown on the bottom-right of Fig.~\ref{fig:performance}. The execution above shows the case where \texttt{await t} is placed in the middle (the await has no effect because \texttt{IO} already finished, and \texttt{Foo} continues to execute). This leads to worse execution time (placing \texttt{await t} after the two code blocks is also worse). These differences in execution times have been confirmed by running the programs on a real machine.

As demonstrated by the examples in Fig.~\ref{fig:performance}, the performance of an asynchronization depends on the execution environment, e.g.,  the overhead of IO operations like http connections and disk access (in Fig.~\ref{fig:performance}, we use \texttt{Thread.Sleep(n)} or \texttt{Task.Delay(n)} to model such overheads). Since modeling the behavior of an execution environment w.r.t.  performance is difficult in general, selecting the most performant asynchronization using static reasoning is also difficult. As a way of sidestepping this difficulty, we focus on enumerating \emph{all} sound asynchronizations that allows to evaluate performance separately in a dynamic manner using performance tests for instance (for each sound asynchronization).

In the worst-case, the number of (sound) asynchronizations is exponential in the number of method calls in the program. Therefore, we focus on the \emph{delay complexity} of the problem of enumerating sound asynchronizations, i.e., the complexity of the delay between outputting two consecutive (distinct) solutions, and show that this is polynomial time modulo an oracle for solving reachability (assertion checking) in \emph{sequential} programs. Note that a trivial enumeration of all asynchronizations and checking equivalence for each one of them has an exponential delay complexity modulo an oracle for checking equivalence. 

As an intermediate step, we consider the problem of computing \emph{maximal} sound asynchronizations that maximize the distance between every call and its matching \texttt{await}. We show that rather surprisingly, 
there exists a \emph{unique} maximal sound asynchronization. This is not trivial since asynchronizations can be incomparable w.r.t. distances between calls and \texttt{await}s (i.e., better for one \texttt{await} and worse for another, and vice-versa).
%
%
This holds even if maximality is relative to a given asynchronization $\Aprog$ imposing an upper bound on the distance between awaits and calls. In principle, avoiding data races could reduce to a choice between moving one await or another closer to the matching call. We show that this is not necessary because the maximal asynchronization is required to be equivalent to a \emph{sequential} program, which 
executes statements in a fixed order. 

As a more pragmatic approach, we define a procedure for computing sound asynchronizations which relies on a bottom-up interprocedural data-flow analysis. The placement of awaits is computed by traversing the call graph bottom up
and using a data-flow analysis that computes read or write accesses made in the callees. We show that this procedure computes maximal sound asynchronizations of abstracted programs where every Boolean condition 
is replaced with non-deterministic choice. These asynchronizations are sound for the concrete programs as well. This procedure enables a polynomial-time delay enumeration of sound asynchronizations of abstracted programs.

%
%

We implemented the asynchronization enumeration based on data-flow analysis in a prototype tool for C\# programs. We evaluated this implementation on a number of non-trivial programs extracted from open source repositories 
to show that our techniques have the potential to become the basis of refactoring tools that allow programmers to improve their usage of async/await primitives.

In summary, this paper makes the following contributions:
\vspace{-1.5mm}
\begin{itemize}
    \item Define the problem of data race-free (sound) asynchronization synthesis for refactoring sequential code to equivalent asynchronous code (Section~\ref{sec:4}). 
    \item Show that the problem of computing a sound asynchronization that maximizes the distance between calls and awaits has a unique solution  (Section~\ref{sec:synthe-corrctProg}).
    \item The delay complexity of sound asynchronization synthesis (Sections~\ref{sec:6}--\ref{sec:asymComplexity}).  
    \item A pragmatic algorithm for computing sound asynchronizations based on a data-flow analysis  (Section~\ref{sec:analysis}).
    \item A prototype implementation of this algorithm and an evaluation of this prototype on a benchmark of non-trivial C\# programs (Section~\ref{sec:experiments}). 
\vspace{-3mm}
\end{itemize}
Additional formalization and proofs are included in the appendix.

\vspace{-5pt}
\section{Asynchronous Programs}\label{sec:1}
\vspace{-3pt}

We consider a simple programming language to formalize our approach, shown in Fig.~\ref{Figure:syntax}. 
A \emph{program} is a set of methods, including a distinguished \texttt{main}, 
which are classified as \emph{synchronous} or \emph{asynchronous}.  
Synchronous methods 
run continuously until completion when they are invoked.
Asynchronous methods, marked using the keyword \texttt{async}, 
can run only partially and be interrupted when executing an \texttt{await}. 
Only asynchronous methods can use \texttt{await}, and all methods 
using \texttt{await} must be defined as asynchronous. We assume that methods are not (mutually) recursive.
A program is called \emph{synchronous} if it is a set of synchronous methods. 

A method is defined by a name from a set $\mathbb{M}$ and 
a list of statements over 
a set $\mathbb{PV}$ of \emph{program variables}, which 
can be accessed from different methods (ranged over using $x$, $y$, $z$,$\ldots$), 
and a set $\mathbb{LV}$ of method \emph{local variables} (ranged over using $r$, $r_1$, $r_2$,$\ldots$).  
Input/return parameters are modeled using program variables. 
Each method call returns a \emph{unique task identifier} from a set $\mathbb{T}$, 
used to record control dependencies imposed by \plog{await}s 
(for uniformity, synchronous methods return a task identifier as well). 
Our language includes assignments, 
\plog{await}s, \plog{return}s, loops, and conditionals. 
Assignments to a local variable $r := x$, where $x$ is a program variable, 
are called \emph{reads} of $ x$, and assignments to a program 
variable $ x :=  le$ ($le$ is an expression over local variables) are called \emph{writes} 
to $x$. A \emph{base} method is a method whose body does \emph{not} contain method calls.

\begin{figure}[t]
  {\footnotesize
  \setlength{\grammarindent}{9em}
  \setlength{\grammarparsep}{0.1cm}
  \begin{grammar}
  <prog> ::= \plog{program}  <md>
  
  <md> ::= \plog{method} <m> \{ <inst> \} |  \plog{async\ method} <m> \{ <inst> \} | <md> <md>
  
  <inst> ::=  <x> ":=" <le> | <r> ":=" <x> | <r> := \plog{call} <m> | \plog{return} | \plog{await} <r> | \plog{await} $*$ | \plog{if} <le> \{<inst>\} \plog{else} \{<inst>\} | \plog{while} <le> \{<inst>\} | <inst> ; <inst>
  \end{grammar}}
  \vspace{-5mm}
  \caption{Syntax. $\langle m\rangle$, $\langle x \rangle$, and $\langle r \rangle$ represent method names, program and local variables, resp. $\langle le \rangle$ is an expression over local variables, or $*$ which is non-deterministic choice.
  }
  \label{Figure:syntax}
  \vspace{-5mm}
  \end{figure}

\noindent
\textbf{Asynchronous methods.} Asynchronous methods can use \plog{await}s to wait for the completion of a task (invocation) while \emph{the control is passed to their caller}. The parameter  $r$ of the \plog{await} specifies the id of the awaited task. As a sound abstraction of awaiting the completion of an IO operation (reading or writing a file, an http request, etc.), which we do not model explicitly, we use a variation \plog{await} $*$. This has a non-deterministic effect of either continuing to the next statement in the same method (as if the IO operation already completed), or passing the control to the caller (as if the IO operation is still pending).  

\begin{wrapfigure}{r}{0.32\textwidth}
\vspace{-1.5cm}
\lstset{basicstyle=\ttfamily\scriptsize,numbers=none,
          stepnumber=1,numberblanklines=false,mathescape=true,morekeywords={method,async,await}}
\hspace{-11mm}
\begin{minipage}[t]{0.47\textwidth}
\begin{lstlisting}  
async method ReadToEndAsync() {       
 await $*$;
 ind = Stream.index;
 len = Stream.content.Length;
 if (ind >= len)
  retVal  =  ""; return
 Stream.index = len;     
 retVal = Stream.content(ind,len);
 return                       }
\end{lstlisting}
\end{minipage}
\vspace{-0.6cm}
\caption{An IO method.}
\label{fig:IO-modeling}
\vspace{-.99cm}
\end{wrapfigure}
Fig.~\ref{fig:IO-modeling} lists our modeling of the IO method \texttt{Read\-ToEnd\-Async} used in Fig.~\ref{fig:example}. We use program variables to represent system resources such as the file system. The await for the completion of accesses to such resources is modeled by \plog{await} $*$.  This enables capturing racing accesses to system resources in asynchronous executions. Parameters or return values are modeled using program variables. \texttt{ReadToEndAsync} is modeled using reads/writes of the index/content of the input stream, and \plog{await} $*$ models the await for their completion.

We assume that the body of every asynchronous method $\ameth$ satisfies several well-formedness 
syntactic constraints, defined on its control-flow graph (CFG). 
We recall that each node of the CFG represents a basic block of code (a maximal-length sequence of branch-free code), and nodes are connected by directed edges which represent a possible transfer of control between blocks. Thus,
\begin{enumerate}[noitemsep,topsep=1pt]
	\item every call $r := \mathtt{call}\ \ameth'$ uses a distinct variable $r$ (to store task identifiers),
	\item every CFG block containing an $\mathtt{await}\ r$ is dominated by the CFG block containing the call $r := \mathtt{call}\ \ldots$ (i.e., every CFG path from the entry to the await has to pass through the call),
	\item every CFG path starting from a block containing a call $r := \mathtt{call}\ \ldots$ to the exit has to pass through an $\mathtt{await}\ r$ statement.
\end{enumerate}
The first condition simplifies the technical exposition, while the last two ensure that $r$ stores a valid task identifier when executing an $\mathtt{await}\ r$, and that every asynchronous invocation is awaited before the caller finishes. Languages like C\# or Javascript do not enforce the latter constraint, but it is considered bad practice due to possible exceptions that may arise in the invoked task and are not caught. 
We forbid passing task identifiers as method parameters (which is possible in C\#). A statement $\mathtt{await}\ r$  is said to \emph{match} a statement $r := \mathtt{call}\ \ameth'$. 

\begin{wrapfigure}{r}{0.59\textwidth}
  \vspace{-1.1cm}
  \lstset{basicstyle=\ttfamily\scriptsize,numbers=none,
            stepnumber=1,numberblanklines=false,mathescape=true,morekeywords={method,async,await}}
  \begin{minipage}[l]{0.323\linewidth}
  \begin{lstlisting}  
async method m {         
  while $*$
    r = call m1;

  await r;
}  
  \end{lstlisting}
  \end{minipage}
  \hfill
  \begin{minipage}{0.323\linewidth}
  \begin{lstlisting} 
async method m {   
  r = call m1; 

  if $*$
    await r;
} 
  \end{lstlisting}
  \end{minipage}
  \hfill
  \begin{minipage}[r]{0.323\linewidth}
  \begin{lstlisting} 
async method m { 
  r = call m1; 
  while $*$
    r' = call m1;
    await r';
  await r;   
} 
  \end{lstlisting}
  \end{minipage}
  \vspace{-0.9cm}
    \caption{Examples of programs} 
    \label{fig:synatactic-examples}
    \vspace{-0.9cm}
  \end{wrapfigure}
In Fig.~\ref{fig:synatactic-examples}, we give three examples of programs to explain in more details the well-formedness 
syntactic constraints. The program on the left of Fig.~\ref{fig:synatactic-examples} does not satisfy the second condition since \texttt{await} \texttt{r} can be reached without entering the loop. The program in the center of Fig.~\ref{fig:synatactic-examples} does not satisfy the third condition since we can reach the end of the method without entering the if branch and thus, without executing \texttt{await} \texttt{r}. The program on the right of Fig.~\ref{fig:synatactic-examples} satisfies both conditions.


\noindent
\textbf{Semantics.}
A program configuration is a tuple $(\gsconf, \rtaskconf, \wtaskconf,\returnconf, \allowbreak \callrel, \awaitrel)$ where $\gsconf$ is composed of the valuation of the program variables excluding the program counter, 
$\rtaskconf$ is the call stack, $\wtaskconf$ is the set of asynchronous tasks, e.g., continuations predicated on the completion of some method call, $\returnconf$ is the set of completed tasks, $\callrel$ represents the relation between a method call and its caller, and $\awaitrel$ represents the control dependencies imposed by \plog{await} statements. 
The activation frames in the call stack and the asynchronous tasks are represented using triples $(\aniden, \ameth,\lsconf)$ where $\aniden\in\mathbb{T}$ is a task identifier, $\ameth\in \mathbb{M}$ is a method name, and $\lsconf$ is a valuation of local variables, including as usual a dedicated program counter.
The set of completed tasks is represented as a function $\returnconf:\mathbb{T}\rightarrow\{\top,\perp\}$ such that $\returnconf(\aniden)=\top$ when $\aniden$ is completed and $\returnconf(\aniden)=\perp$, otherwise.
We define $\callrel$ and $\awaitrel$ as partial functions $\mathbb{T}\rightharpoonup \mathbb{T}$ with the meaning that $\callrel(\aniden)=\anidenp$, resp., $\awaitrel(\aniden)=\anidenp$, iff $\aniden$ is called by $\anidenp$, resp., $\aniden$ is waiting for $\anidenp$. 
We set $\awaitrel(\aniden)=*$ if the task $\aniden$ was interrupted because of an $\mbox{\plog{await}}\ *$ statement. 

The semantics of a program $\aprog$ is defined as a labeled transition system (LTS) $[\aprog] =(\pstatesconf,\actionsconf,\psconf_0,\rightarrow)$ where $\pstatesconf$ is the set of program configurations, $\actionsconf$ is a set of transition labels called \emph{actions}, $\psconf_0$ is the initial configuration, and $\rightarrow\subseteq \pstatesconf\times \actionsconf\times \pstatesconf$ is the transition relation. Each program statement is interpreted as a transition in $[\aprog]$. The set of actions is defined by ($\actionsid$ is a set of action identifiers):
\vspace{-2mm}
\begin{align*}
  \actionsconf = & \{ (\actiden, \aniden, \event): \actiden \in \actionsid, \aniden \in \mathbb{T}, \event\in\{ \loadact(\anaddr), \storeact(\anaddr), \callact(\anidenp), \awaitact(\anidenk), \returnact,   \\ & \continueact: \anidenp \in \mathbb{T}, \anidenk \in \mathbb{T}\cup\{*\}, \anaddr\in \mathbb{PV}\}\} 
\vspace{-3mm}
\end{align*}

The transition relation $\rightarrow$ is defined in Fig.~\ref{Table:SemanticsRules}. Transition labels are written on top of $\rightarrow$. 

Transitions labeled by $(\actiden, \aniden, \loadact(\anaddr))$ and $(\actiden, \aniden, \storeact(\anaddr))$ represent a read and a write accesses to the program variable $\anaddr$, respectively, executed by the task (method call) with identifier $\aniden$. 
A transition labeled by $(\actiden, \aniden, \callact(\anidenp))$ corresponds to the fact that task $\aniden$ executes a method call that results in creating a task $\anidenp$.
Task $\anidenp$ is added on the top of the stack of currently executing tasks, declared pending (setting $\returnconf(j)$ to $\perp$), and $\callrel$ is updated to track its caller ($\callrel(\anidenp)=\aniden$).
A transition $(\actiden, \aniden, \returnact)$ represents the return from task $\aniden$.
Task $\aniden$ is removed from the stack of currently executing tasks, and $\returnconf(\aniden)$ is set to $\top$ to record the fact that task $\aniden$ is finished.

A transition $(\actiden, \aniden, \awaitact(\anidenp))$ relates to task $\aniden$ waiting asynchronously for task $\anidenp$. Its effect depends on whether task $\anidenp$ is already completed. If this is the case (i.e., $\returnconf[\anidenp]=\top$), task $\aniden$ continues and executes the next statement.
Otherwise, task $\aniden$ executing the $\awaitact$ is removed from the stack and added to the set of pending tasks, and $\awaitrel$ is updated to track the waiting-for relationship ($\awaitrel(\aniden)=\anidenp$). 
Similarly, a transition $(\actiden, \aniden, \awaitact(*))$ corresponds to task $\aniden$ waiting asynchronously for the completion of an unspecified task. Non-deterministically, task $\aniden$ continues to the next statement, or task $\aniden$ is interrupted and transferred to the set of pending tasks ($\awaitrel(\aniden)$ is set to $*$).

A transition $(\actiden, \aniden, \continueact)$ represents the scheduling of the continuation of task $\aniden$. There are two cases depending on whether $\aniden$ waited for the completion of another task $\anidenp$ modeled explicitly in the language (i.e., $\awaitrel(\aniden)=\anidenp$), or an unspecified task (i.e., $\awaitrel(\aniden)=*$). In the first case, the transition is enabled only when the call stack is empty and $\anidenp$ is completed. In the second case, the transition is always enabled. The latter models the fact that methods implementing IO operations (waiting for unspecified tasks in our language) are executed in background threads and can interleave with the main thread (that executes the \texttt{Main} method). Although this may seem restricted because we do not allow arbitrary interleavings between IO methods and \texttt{Main}, this is actually sound when focusing on the existence of data races as in our approach. As shown later in Table~\ref{Table:TracesOrders}, any two instructions that follow an \plog{await} $*$ are not happens-before related and form a race.

By the definition of $\rightarrow$, every action $\action\in \actionsconf \setminus \{ (\_, \_, \continueact)\}$ corresponds to executing some statement in the program, which is denoted by $\acttost(a)$.

\begin{figure}[!t]
  \small\addtolength{\tabcolsep}{-5pt}
  \therules{
  \scriptsize
  \therule
  {$\text{\plog{\theload{\areg}{\anaddr}}}\in\instrOf(\lsconf(\pcconf))$\quad
  $\actiden\in\mathbb{\actionsid}$ fresh \quad 
  $\lsconf' = \lsconf[\areg \mapsto \gsconf(\anaddr),\pcconf\mapsto \mathsf{next}(\lsconf(\pcconf))]$}
  {$(\gsconf, (\aniden,\ameth,\lsconf) \circ \rtaskconf, \_, \_, \_, \_)
  \mpitrans{(\actiden, \aniden, \loadact(\anaddr))}
   (\gsconf, (\aniden,\ameth,\lsconf') \circ \rtaskconf, \_, \_, \_, \_)$}
  \scriptsize
  \therule
  {$\text{\plog{\thestore{\anaddr}{le}}}\in\instrOf(\lsconf(\pcconf))$\quad
  $\actiden\in\mathbb{\actionsid}$ fresh \quad 
  $\lsconf' = \lsconf[\pcconf\mapsto \mathsf{next}(\lsconf(\pcconf))]$ \quad
  $\gsconf' = \gsconf[\anaddr \mapsto \lsconf(\text{\plog{le}})]$}
  {$(\gsconf, (\aniden,\ameth,\lsconf) \circ \rtaskconf, \_, \_, \_, \_)
  \mpitrans{(\actiden, \aniden, \storeact(\anaddr))}
  (\gsconf', (\aniden,\ameth,\lsconf') \circ \rtaskconf, \_, \_, \_, \_)$}
  \scriptsize
  \dfrac
  {\splitdfrac{\text{$\theassign{\areg}{\callact\ \ameth}\in\instrOf(\lsconf(\pcconf))$\quad
  $\actiden\in\mathbb{\actionsid}$ fresh \quad 
  $\ell_0 = \mathsf{init}(\gsconf,\ameth)$ \quad
  $\anidenp\in\mathbb{T}$ fresh  }}
  {\text{$\ell' = \lsconf[\areg\mapsto\anidenp,\pcconf\mapsto \mathsf{next}(\lsconf(\pcconf))]$ \quad 
  $\returnconf' = \returnconf[\anidenp \mapsto \perp]$ \quad $\callrel' = \callrel[\anidenp \mapsto \aniden]$}}}
  {\text{$(\gsconf, (\aniden, \ameth',\lsconf) \circ \rtaskconf, \_, \returnconf, \callrel, \_)
  \mpitrans{(\actiden, \aniden, \callact(\anidenp))}
  (\gsconf, (\anidenp,\ameth,\lsconf_0) \circ (\aniden,\ameth',\ell') \circ \rtaskconf, \_, \returnconf', \callrel', \_)$}}
  \\[1.5em]
  \scriptsize
  \therule
  {$\text{\plog{\returnact}}\in\instrOf(\lsconf(\pcconf))$ \quad
  $\actiden\in\mathbb{\actionsid}$ fresh \quad 
  $\returnconf' = \returnconf[\aniden \mapsto \top]$}
  {$(\gsconf, (\aniden, \ameth,\lsconf) \circ \rtaskconf, \_, \returnconf, \_, \_)
  \mpitrans{(\actiden, \aniden, \returnact)}
  (\gsconf, \rtaskconf, \_, \returnconf', \_, \_)$}
  \scriptsize
  \therule
  {$\text{\plog{\awaitact\ \areg}}\in\instrOf(\lsconf(\pcconf))$ \quad
  $\actiden\in\mathbb{\actionsid}$ fresh \quad 
  $\returnconf(\lsconf(\areg)) = \top$ \quad
  $\lsconf' = \lsconf[\pcconf\mapsto \mathsf{next}(\lsconf(\pcconf))]$ }
  {$(\gsconf, (\aniden, \ameth,\lsconf) \circ \rtaskconf, \_, \returnconf, \_, \_)
  \mpitrans{(\actiden, \aniden, \awaitact(\lsconf(\areg)))}
  (\gsconf, (\aniden, \ameth, \lsconf') \circ \rtaskconf, \_, \returnconf, \_, \_)$}
  \scriptsize
  \dfrac
  {\splitdfrac{\text{$\text{\plog{\awaitact\ \areg}}\in\instrOf(\lsconf(\pcconf))$ \quad
  $\actiden\in\mathbb{\actionsid}$ fresh \quad 
  $\returnconf(\lsconf(\areg)) = \perp$ \quad
  $\awaitrel' = \awaitrel[\aniden \mapsto \lsconf(\areg)]$}}
  {\text{$\lsconf' = \lsconf[\pcconf\mapsto \mathsf{next}(\lsconf(\pcconf))]$}}}
  {\text{$(\gsconf, (\aniden, \ameth,\lsconf) \circ \rtaskconf, \wtaskconf, \returnconf, \_, \awaitrel)
  \mpitrans{(\actiden, \aniden, \awaitact(\lsconf(\areg)))}
  (\gsconf, \rtaskconf, \{(\aniden, \ameth,\lsconf')\} \uplus \wtaskconf, \returnconf, \_, \awaitrel')$}}
  \\[1.5em]
  \scriptsize
  \therule
  {$\text{\plog{\awaitact\ *}}\in\instrOf(\lsconf(\pcconf))$ \quad
  $\actiden\in\mathbb{\actionsid}$ fresh \quad 
  $\lsconf' = \lsconf[\pcconf\mapsto \mathsf{next}(\lsconf(\pcconf))]$ }
  {$(\gsconf, (\aniden, \ameth,\lsconf) \circ \rtaskconf, \_, \_, \_, \_)
  \mpitrans{(\actiden, \aniden, \awaitact(*))}
  (\gsconf, (\aniden, \ameth, \lsconf') \circ \rtaskconf, \_, \_, \_, \_)$}
  \scriptsize
  \therule
  {$\text{\plog{\awaitact\ *}}\in\instrOf(\lsconf(\pcconf))$ \quad
  $\actiden\in\mathbb{\actionsid}$ fresh \quad 
  $\awaitrel' = \awaitrel[\aniden \mapsto \ *]$ \quad
  $\lsconf' = \lsconf[\pcconf\mapsto \mathsf{next}(\lsconf(\pcconf))]$}
  {$(\gsconf, (\aniden, \ameth,\lsconf) \circ \rtaskconf, \wtaskconf, \_, \_, \awaitrel)
  \mpitrans{(\actiden, \aniden, \awaitact(*))}
  (\gsconf, \rtaskconf, \{(\aniden, \ameth,\lsconf')\} \uplus \wtaskconf, \_, \_, \awaitrel')$}
  \scriptsize
  \therule
  {$\actiden\in\mathbb{\actionsid}$ fresh \quad 
  $\awaitrel(\aniden) = \anidenp$ \quad
  $\returnconf(\anidenp) = \top$
  }
  {$(\gsconf, \epsilon, \{(\aniden, \ameth,\lsconf)\} \uplus \wtaskconf, \returnconf, \_, \awaitrel)
  \mpitrans{(\actiden, \aniden, \continueact)}
  (\gsconf,  (\aniden, \ameth,\lsconf) ,  \wtaskconf, \returnconf, \_,  \awaitrel)$}
  \scriptsize
  \therule
  {$\actiden\in\mathbb{\actionsid}$ fresh \quad 
  $\awaitrel(\aniden) = *$
  }
  {$(\gsconf, \rtaskconf, \{(\aniden, \ameth,\lsconf)\} \uplus \wtaskconf, \_, \_, \awaitrel)
  \mpitrans{(\actiden, \aniden, \continueact)}
  (\gsconf,  (\aniden, \ameth,\lsconf) \circ \rtaskconf,  \wtaskconf, \_, \_,  \awaitrel)$}
  }
  \vspace{-0.7cm}
  \caption{Program semantics. For a function $f$, we use $f[a\mapsto b]$ to denote a function $g$ such that $g(c)=f(c)$ for all $c\neq a$ and $g(a)=b$. The function $\instrOf$ returns the instruction at some given control location while $\mathsf{next}$ gives the next instruction to execute. We use $\circ$ to denote sequence concatenation and $\mathsf{init}$ to denote the initial state of a method call.
  }
  \label{Table:SemanticsRules}
  \vspace{-1cm}
  \end{figure}

An execution of $\aprog$ is a sequence $\rho=\psconf_0\xrightarrow{\action_1}\psconf_1\xrightarrow{\action_2}\ldots$ of transitions starting in the initial configuration $\psconf_0$
  and leading to a configuration $\psconf$ where the call stack and the set of pending tasks are empty. 
   $\pstatesconf[\aprog]$ denotes the set of all program variable valuations included in configurations that are reached in executions of $\aprog$.
  Since we are only interested in reasoning about the sequence of actions $\action_1\cdot\action_2\cdot\ldots$ labeling the transitions of an execution, we will call the latter an execution as well. The set of executions of a program $\aprog$ is denoted by $\executionsconf(\aprog)$. 
  
  \noindent
  {\bf Traces.} 
  The \emph{trace} of an execution $\rho \in \executionsconf(\aprog)$ is a tuple $\traceof{\rho} = (\rho, \mo{}, \co{}, \so{}, \hbo{})$ of strict partial orders between the actions in $\rho$ defined in Table \ref{Table:TracesOrders}. The \emph{method invocation order} $\mo$ records the order between actions in the same invocation, and the \emph{call order} $\co$ is an extension of $\mo$ that additionally orders actions before an invocation with respect to those inside that invocation. The \emph{synchronous happens-before order} $\so$ orders the actions in an execution as if all the invocations were synchronous (even if the execution may contain asynchronous ones). It is an extension of $\co$ where additionally, every action inside a callee is ordered before the actions following its invocation in the caller. The (asynchronous) \emph{happens-before order} $\hbo{}$ contains typical control-flow constraints: it is an extension of $\co$ where every action $a$ inside an asynchronous invocation is ordered before the corresponding $\mathtt{await}$ in the caller, and before the actions following its invocation in the caller if $a$ precedes the first\footnote{Code in between two awaits can execute before or after the control is returned to the caller, depending on whether the first awaited task finished or not.} $\mathtt{await}$ in $\mo$ (an invocation can be interrupted only when executing an $\mathtt{await}$) or if the callee does not contain an $\mathtt{await}$ (it is  synchronous).
  $\tracesconf(\aprog)$ is the set of traces of $\aprog$.

  \begin{table}
    \vspace{-7mm}
    \caption{Strict partial orders included in a trace. $\co{}$, $\so{}$, and $\hbo{}$ are the smallest satisfying relations.} \label{Table:TracesOrders}
    \vspace{-0mm}
  {\footnotesize
  \begin{tabular}{|c|c|}
    \hline
    $\action_1 <_{\rho} \action_2$ & $\action_1$ occurs before  $\action_2$ in $\rho$ and $\action_1 \neq \action_2$\\[.5mm] \hline
    $\action_1 \sim \action_2$ & $\action_1=(\_,\aniden, \_)$ and $\action_2=(\_,\aniden, \_)$ \\[.5mm] \hline
    \hline
    $(\action_1,\action_2) \in \mo{}$ & $\action_1 \sim \action_2 \wedge \action_1 <_{\rho} \action_2$ \\ \hline
    $(\action_1,\action_2) \in \co{}$ & $(\action_1,\action_2) \in \mo{} \vee (\action_1 = (\_,\aniden, \callact(\anidenp)) \wedge \action_2 = (\_,\anidenp,\_))$ \\
    & $\vee\ (\exists\ \action_3.\  (\action_1,\action_3) \in \co{} \wedge (\action_3,\action_2) \in \co{})$ \\[.5mm]\hline
    $(\action_1,\action_2) \in \so{}$ & $(\action_1,\action_2) \in \co{} \vee (\exists\ \action_3.\  (\action_1,\action_3) \in \so{} \wedge (\action_3,\action_2) \in \so{})$ \\
    & $\vee\ (\action_1 = (\_,\anidenp, \_) \wedge \action_2 = (\_,\aniden,\_) \wedge \exists\ \action_3 = (\_,\aniden, \callact(\anidenp)).\ \action_3 <_{\rho} \action_2)$ \\[.5mm]\hline
    $(\action_1,\action_2) \in \hbo{}$  & $(\action_1,\action_2) \in \co{} \vee (\exists\ \action_3.\  (\action_1,\action_3) \in \hbo{} \wedge (\action_3,\action_2) \in \hbo{})$ \\ 
     & $\vee\ (\ \action_1 = (\_,\anidenp, \_) \wedge \action_2 = (\_,\aniden,\_) \wedge \exists\ \action_3 = (\_,\aniden, \awaitact(\anidenp)).\ \action_3 <_{\rho} \action_2\ )$ \\
      & $\vee\ (\ \action_1 = (\_, \anidenp, \awaitact(\aniden'))$ is the first await in $\anidenp\ \wedge$\\
     & $\action_2 = (\_,\aniden,\_)  \wedge \exists\ \action_3 = (\_,\aniden, \callact(\anidenp)).\ \action_3 <_{\rho} \action_2\ )$ \\
     & $\vee\ (\ \action_1 = (\_, \anidenp, \_)\ \wedge \not\exists\ (\_, \anidenp, \awaitact(\_)) \in\rho\ \wedge$ \\ 
     & $\action_2 = (\_,\aniden,\_)  \wedge \exists\ \action_3 = (\_,\aniden, \callact(\anidenp)).\ \action_3 <_{\rho} \action_2\ )$ \\ [.5mm]\hline
  \end{tabular}}
  \vspace{-5mm}
  \end{table}

On the right of Fig.~\ref{fig:example}, we show a trace where two statements (represented by the corresponding lines numbers) are linked by a dotted arrow if 
the corresponding actions are related by $\mo$, a dashed arrow if the corresponding actions are related by 
 $\co$ but not by $\mo$, and a solid arrow if the corresponding actions are related by the $\hbo{}$ but not by $\co$.

\vspace{-5pt}
\section{Synthesizing Asynchronous Programs}\label{sec:4}
\vspace{-5pt}
Given a synchronous program $\aprog$ and a subset of \emph{base} methods $\alib\subseteq \aprog$,
our goal is to synthesize \emph{all} asynchronous programs $\aprog_a$ that are equivalent to $\aprog$ and that are obtained by substituting every method in $\alib$ with an equivalent \emph{asynchronous} version. The base methods are considered to be models of standard library calls (e.g., IO operations) 
and asynchronous versions are defined by inserting \plog{await} $*$ statements in their body. 
We use $\aprog[\alib]$ to emphasize a subset of base methods $\alib$ in a program $\aprog$. Also, we call $\alib$ a \emph{library}. A library is called (a)synchronous when all methods are (a)synchronous.



\noindent
\textbf{Asynchronizations of a synchronous program.} Let $\aprog[\alib]$ be a synchronous program, and $\alib_a$ a set of asynchronous methods obtained from those in $\alib$ by inserting at least one \plog{await} $*$ statement in their body (and adding the keyword \texttt{async}). Each method in $\alib_a$ corresponds to a method in $\alib$ with the same name, and vice-versa. $\Aprog[\Alib]$ is called an \emph{asynchronization} of $\aprog[\alib]$ with respect to $\alib_a$  
if it is a syntactically correct program obtained by replacing the methods in $\alib$ with those in $\alib_a$ and adding \plog{await} statements as necessary. 
\begin{wrapfigure}{r}{0.57\textwidth}
  \vspace{-0.9cm}
\lstset{basicstyle=\ttfamily\scriptsize,numbers=none,
			  stepnumber=1,numberblanklines=false,mathescape=true,morekeywords={method,async,await}}
\begin{minipage}[l]{0.31\linewidth}
\begin{lstlisting}  
method m {  
  r1 = call m1;

  r2 = x;
            }
method m1 {

  retVal = x;
  x = input;
  return;  }
\end{lstlisting}
\end{minipage}
\hspace{-3mm}
\begin{minipage}[l]{0.34\linewidth}
\begin{lstlisting} 
async method m {  
  r1 = call m1;
  await r1;
  r2 = x;
                 } 
async method m1 {
  await $*$
  retVal = x;
  x = input;
  return;      }
\end{lstlisting}
\end{minipage}
\hspace{0mm}
\begin{minipage}[r]{0.34\linewidth}
\begin{lstlisting} 
async method m {  
  r1 = call m1;

  r2 = x;
  await r1;      } 
async method m1 {
  await $*$
  retVal = x;
  x = input;
  return;      }
\end{lstlisting}
\end{minipage}
\vspace{-0.4cm}
\caption{A program and its asynchronizations.}
\label{fig:asynchronization}
\vspace{-0.6cm}
\end{wrapfigure}  
More precisely, let $\alib^*\subseteq \aprog$ be the set of all methods of $\aprog$ that transitively call methods of $\alib$. Formally, $\alib^*$ is the smallest set of methods that includes $\alib$ and satisfies the following: if a method $\ameth$ calls $\ameth'\in \alib^*$, then $\ameth \in \alib^*$. Then, $\Aprog[\Alib]$ is an \emph{asynchronization} of $\aprog[\alib]$ w.r.t.  $\alib_a$ if it is obtained from $\aprog$ as follows:
\begin{itemize}[noitemsep,topsep=0pt]
  \item Each method in $\alib$ is replaced with the corresponding method from $\alib_a$.
	\item All methods in $\alib^*\setminus \alib$ are declared as asynchronous (because every call to an asynchronous method is followed by an \plog{await} and any method using \plog{await} must be asynchronous). 
	\item For each invocation $\theassign{\areg}{\callact\ \ameth}$ of $\ameth\in \alib^*$, add \texttt{aw\-a\-it} statements $\awaitact\ \areg$ satisfying the well-formedness syntactic constraints described in Section~\ref{sec:1}. 
\end{itemize}
Fig. \ref{fig:asynchronization} lists a synchronous program and its two asynchronizations, wh\-ere $\alib=\{\ameth 1\}$ and  $\alib^*=\{\ameth, \ameth 1\}$. 
Asynchronizations differ only in the await placement.
 
$\Asyof{\aprog, \alib, \Alib}$ is the set of all asynchronizations of $\aprog[\alib]$ w.r.t. $\alib_a$. 
The \emph{strong} asynchronization $\BAsyof{\aprog, \alib, \Alib}$ is an asynchronization where every \plog{await} \emph{immediately} follows the matching call. It reaches exactly the same set of program variable valuations as $\aprog$.


\noindent
\textbf{Problem definition.} We investigate the problem of enumerating \emph{all} asynchronizations of a given program w.r.t. a given asynchronous library, which are \emph{sound}, in the sense that they do not admit data races. 
Two actions $\action_1$ and $\action_2$ in a trace $\tau= (\rho, \mo{}, \co{}, \so{}, \hbo{})$ are \emph{concurrent} if $(\action_1, \action_2) \not\in \hbo{}$ and $(\action_2, \action_1) \not\in \hbo{}$.

An ansynchronous program $\Aprog$ \emph{admits a data race $(\action_1, \action_2)$}, where $(\action_1, \action_2)\in \so{}$, if $\action_1$ and $\action_2$ are two concurrent actions of a trace $\tau\in \tracesconf(\Aprog)$, and $\action_1$ and $\action_2$ are read or write accesses to the same program variable $\anaddr$, and at least one of them is a write.
%
We write data races as ordered pairs w.r.t. $\so{}$ to simplify the definition of the algorithms in the next sections. 
Also, note that traces of \emph{synchronous} programs can \emph{not} contain concurrent actions, and therefore they do not admit data races. $\BAsyof{\aprog, \alib, \Alib}$ does not admit data races as well.

$\Aprog[\Alib]$ is called \emph{sound} when it does not admit data races.
%
The absence of data races implies equivalence to the original program, in the sense of reaching the same set of configurations (program variable valuations).

\vspace{-2mm}
\begin{definition}\label{def:problem1}
  For a synchronous program $\aprog[\alib]$ and asynchronous library $\Alib$, the \emph{asychronization synthesis problem} asks to enumerate all sound asynchronizations in $\Asyof{\aprog, \alib, \Alib}$. 
  \vspace{-3mm}
\end{definition}


\vspace{-5pt}
\section{Enumerating Sound Asynchronizations} \label{sec:synthe-corrctProg}
\vspace{-5pt}

We present an algorithm for solving asynchronization synthesis, which relies on a partial order between asynchronizations that guides the enumeration of possible solutions. The partial order takes into account the distance between calls and corresponding awaits. 
Fig. \ref{fig:space} pictures the partial order for asynchronizations of the program on the left of Fig.~\ref{fig:example}. Each asynchronization is written as 
\begin{wrapfigure}{r}{0.15\textwidth}
  \vspace{-7mm}    
  \begin{minipage}{0.15\textwidth}
          \scalebox{0.78} {
              \begin{tikzpicture}[node distance=2.5cm]
              \node(A0)                           {$(2,1)$};
              \node(B0)      [below=0.35cm of A0]  {};
              \node(B2)      [right=0.015cm of B0, fill=yellow!30]        {$(1,1)$};
              \node(B1)      [left=0.015cm of B0]         {$(2,0)$};
              \node(C0)      [below=0.35cm of B1, fill=yellow!30]    {$(1,0)$};
              \node(C1)      [below=0.35cm of B2, fill=yellow!30]  {$(0,1)$};
              \node(D1)      [below=0.35cm of C0]  {};
              \node(D0)      [right=0.015cm of D1, fill=yellow!30]  {$(0,0)$};
          
              \draw(A0)       --  (B1);
              \draw(A0)       --  (B2);
              \draw(B2)       --  (C1);
              \draw(B2)       --  (C0);
              \draw(B1)       --  (C0);
              \draw(C0)       --  (D0);
              \draw(C1)       --  (D0);
             \end{tikzpicture}}
              \vspace{-1mm}
  \end{minipage}
\vspace{-0.5cm}
\caption{}
\label{fig:space}
\vspace{-0.9cm}
\end{wrapfigure}
a vector of distances, the first (second) element is the number of statements between \texttt{await t1} (\texttt{await t}) and the matching call (we count only statements that appear in the sequential program). The edges connect comparable elements, smaller elements being below bigger elements. The asynchronization on the middle of Fig.~\ref{fig:example} corresponds to the vector $(1,1)$. 
 The highlighted elements constitute the set of all sound asynchronizations. 
 The strong asynchronization corresponds to the vector $(0,0)$.

%
%
Formally, an \plog{await} statement $\statement_w$ in a method $\ameth$ of an asynchronization $\Aprog[\Alib]\in \Asyof{\aprog, \alib, \Alib}$ \emph{covers} a read/write statement $\statement$ in $\aprog$ if there exists a path in the CFG of $\ameth$ 
 from the call statement matching $\statement_w$ to $\statement_w$ that contains $\statement$. The set of statements covered by an await $\statement_w$ is denoted by $\cover(\statement_w)$. 
We compare asynchronizations in terms of sets of statements covered by awaits that match the same call from the synchronous program $\aprog[\alib]$. Since asynchronizations are obtained by adding \plog{await}s, every call in asynchronization $\Aprog[\Alib]\in \Asyof{\aprog, \alib, \Alib}$ corresponds to a \emph{fixed} call in $\aprog[\alib]$. 
Therefore, for two asynchronizations $\Aprog,\Aprog'\in \Asyof{\aprog, \alib, \Alib}$, $\Aprog$ is \emph{smaller} than $\Aprog'$, denoted by $\Aprog \leq \Aprog'$, iff for every \plog{await} $\statement_w$ in $\Aprog$, there exists an \plog{await} $\statement_w'$ in $\Aprog'$ that matches the same call as $\statement_w$, such that $\cover(\statement_w)\subseteq \cover(\statement_w')$. 
For example, the two asynchronous programs in Fig. \ref{fig:asynchronization} are ordered by $\leq$ since  $\cover(\awaitact\ \areg 1) = \{\}$ in the first and $\cover(\awaitact\ \areg 1) = \{\texttt{r2 = x}\}$ in the second. 
Note that the strong asynchronization is smaller than every other asynchronization. Also, note that $\leq$ has a unique maximal element that is called the weakest asynchronization and denoted by $\WAsyof{\aprog, \alib, \Alib}$. In Fig.~\ref{fig:space}, the weakest asynchronization corresponds to the vector $(2,1)$. 

In the following, we say \emph{moving an await down (resp., up)} when moving the await further away from  (resp. closer to) the matching call while preserving well-formedness conditions in Section \ref{sec:1}. Further away or closer to means increasing or decreasing the set of statements that are covered by the await. For instance, if an await $\statement_w$ in a program $\Aprog$ is preceded by a while loop, then \emph{moving it up} means moving it before the whole loop and not inside the loop body. Otherwise, the third well-formedness condition would be violated. 

\noindent
{\bf Relative Maximality.}
A crucial property of this partial order is that for every asynchronization $\Aprog$, there exists a \emph{unique} maximal asynchronization that is smaller than $\Aprog$ and that is sound. Formally, 
an asynchronization $\Aprog'$ is called a \emph{maximal asynchronization of $\aprog$ relative to $\Aprog$} if 
  (1) $\Aprog' \leq \Aprog$, $\Aprog'$ is sound, and 
  (2) 
  $\forall\ \Aprog'' \in \Asyof{\aprog, \alib, \Alib}.\ \Aprog''$ is sound and $\Aprog'' \leq \aprog_a \Rightarrow \Aprog'' \leq \Aprog'$.

\vspace{-2mm}
\begin{lemma}\label{lemma:optimalAsync}
  Given an asynchronization $\Aprog \in \Asyof{\aprog, \alib, \Alib}$,
  there exists a unique program $\Aprog'$ that is a maximal asynchronization of $\aprog$ relative to $\Aprog$.
  \vspace{-2mm}
\end{lemma}

The asynchronization $\Aprog'$ exists because the bottom element of $\leq$ is sound. To prove uniqueness, assume by contradiction that there exist two incomparable maximal asynchronizations $\Aprog^1$ and $\Aprog^2$ and select the first await $\statement_{w}^{1}$ w.r.t. the control-flow of the sequential program that is placed in different positions in the two programs. Assume that $\statement_{w}^{1}$ is closer to its matching call in $\Aprog^1$. Then, we move $\statement_{w}^{1}$ in $\Aprog^1$ further away from its matching call to the same position as in $\Aprog^2$. This modification does not introduce data races since $\Aprog^2$ is data race free. Thus, the resulting program is data race free, bigger than $\Aprog^1$, and smaller than $\Aprog$ w.r.t. $\leq$ contradicting the fact that $\Aprog^1$ is a maximal asynchronization.

\vspace{-5pt}
\subsection{Enumeration Algorithm}
\vspace{-1pt}
\setlength{\textfloatsep}{10pt}
\begin{algorithm}[t]
  \caption{An algorithm for enumerating all sound asynchronizations (these asynchronizations are obtained as a result of the \textbf{output} instruction). 
  \textsc{MaxRel} returns the maximal asynchronization of $\aprog$ relative to $\Aprog$
  }\label{algo0}
  \begin{algorithmic}[1]
  \Procedure{AsySyn}{$\Aprog$, \textcolor{blue}{$\statement_{w}$}}
  \State $\ \ \Aprog' \leftarrow \textsc{MaxRel}(\Aprog)$; 
  \State \ \ \textbf{output} $\Aprog'$;
  \State $\ \ \mathcal{P}  \leftarrow \NextEle(\Aprog', \textcolor{blue}{\statement_{w}})$;
  \State \ \ \textbf{for each}\ $(\Aprog'', \textcolor{blue}{\statement_{w}''}) \in \mathcal{P} $
  \State $\ \ \ \ \ \ \ \textsc{AsySyn}(\Aprog'',\textcolor{blue}{\statement_{w}''})$;
  \EndProcedure
  \end{algorithmic}
\end{algorithm}

Our algorithm for enumerating all sound asynchronizations is given in Algorithm~\ref{algo0} as a recursive procedure \textsc{AsySyn} that we describe in two phases. 

First, ignore the second argument of \textsc{AsySyn} (in blue), which represents an \plog{await} statement. For an asynchronization $\Aprog$, \textsc{AsySyn} outputs \emph{all} sound asynchronizations that are smaller than $\Aprog$. It uses \textsc{MaxRel} to compute the maximal asynchronization $\Aprog'$ of $\aprog$ relative to $\Aprog$, and then, calls itself recursively for all immediate predecessors of $\Aprog'$. \textsc{AsySyn} outputs all sound asynchronizations of $\aprog$ when given as input the weakest asynchronization of $\aprog$. 

\begin{wrapfigure}{r}{0.44\textwidth}
  \vspace{-1cm}
\lstset{basicstyle=\ttfamily\scriptsize,numbers=none,
			  stepnumber=1,numberblanklines=false,mathescape=true,morekeywords={method,async,await}}
\begin{minipage}[l]{0.47\linewidth}
\begin{lstlisting} 
async method m {  
  r1 = call m1;
  r2 = x;
  await r1;     } 
async method m1 {
  r3 = call m2;
  x = x + 1;
  await r3;     }
async method m2 {
  await $*$
  retVal = input;
  return;       }
\end{lstlisting}
\end{minipage}
\hfill
\begin{minipage}[r]{0.47\linewidth}
\begin{lstlisting} 
async method m {  
  r1 = call m1;
  r2 = x;
  await r1;     } 
async method m1 {
  r3 = call m2;
  await r3; 
  x = x + 1;    }
async method m2 {
  await $*$
  retVal = input;
  return;       }
\end{lstlisting}
\end{minipage}
\vspace{-0.4cm}
\caption{Asynchronizations.}
\label{fig:soundasynchronization}
\vspace{-0.6cm}
\end{wrapfigure}
Recursive calls on immediate predecessors are necessary because the set of sound asynchronizations is not downward-closed w.r.t. $\leq$. For instance, the asynchronization on the right of Fig.~\ref{fig:soundasynchronization} is an immediate predecessor of the sound asynchronization on the left but it has a data race on $x$.

The delay complexity of this algorithm remains exponential in general, since a sound asynchronization may be outputted multiple times. Asynchronizations are only partially ordered by $\leq$ and different chains of recursive calls starting in different immediate predecessors may end up outputting the same solution. For instance, for the asynchronizations in Fig.~\ref{fig:space}, the asynchronization $(0,0)$ will be outputted twice because it is an immediate predecessor of both $(1,0)$ and $(0,1)$. 

To avoid this redundancy, we use a refinement of the above that \emph{restricts} the set of immediate predecessors available for a (recursive) call of \textsc{AsySyn}. This is based on a \emph{strict total order} $\wao$ between \plog{await}s in a program $\Aprog$ that follows a topological ordering of its inter-procedural CFG, i.e., if $\statement_{w}$ occurs before  $\statement_{w}'$ in the body of a method $\ameth$, then $\statement_{w}\ \wao\ \statement_{w}'$, and if $\statement_{w}$ occurs in a method $\ameth$ and $\statement_{w}'$ occurs in a method $\ameth'$ s.t. $\meth$ (indirectly) calls $\ameth'$, then $\statement_{w}\ \wao\ \statement_{w}'$. Therefore, \textsc{AsySyn} takes an await statement $\statement_{w}$ as a second parameter, which is initially the maximal element w.r.t. $\wao$, and it calls itself only on immediate predecessors of a solution obtained by \emph{moving up} an await $\statement_{w}''$ \emph{smaller than or equal to} $\statement_{w}$ w.r.t. $\wao$. The recursive call on that predecessor will receive as input $\statement_{w}''$. Formally, this relies on a function $\NextEle$ that returns pairs of immediate predecessors and await statements defined as follows:

\vspace{-4mm}
{\small
\begin{align*} 
  &\NextEle(\Aprog',\textcolor{blue}{\statement_{w}}) = \{ (\Aprog'',\textcolor{blue}{\statement_{w}''}): 
  \Aprog'' < \Aprog'\mbox{ and }
	\forall\ \Aprog''' \in \Asyof{\aprog, \alib, \Alib}.\   \Aprog''' < \Aprog' \implies \Aprog'''  \leq \Aprog''  \\
  &\hspace{4.4cm}  \textcolor{blue}{\mbox{and }\statement_{w}''\ \waoeq\ \statement_{w}\mbox{ and } \Aprog'' \in \Aprog' \uparrow \statement_{w}''}\ \}
\end{align*}}

\vspace{-2mm}
\noindent
($\Aprog' \uparrow \statement_{w}''$ is the set of asynchronizations obtained from $\Aprog'$ by changing \emph{only} the position of $ \statement_{w}''$, moving it up w.r.t. the position in $\Aprog'$).
For instance, looking at immediate predecessors of $(1,1)$ in Fig.~\ref{fig:space}, $(0,1)$ is obtained by moving the \emph{first} await in $\wao$. Therefore, the recursive call on $(0,1)$ computes the maximal asynchronization relative to $(0,1)$, which is $(0,1)$, and stops ($\NextEle$ returns $\emptyset$ because the input $\statement_{w}$ is the minimal element of $\wao$, and already immediately after the call). Its immediate predecessor is explored when recursing on $(1,0)$. 

Algorithm~\ref{algo0} 
outputs all sound asynchronizations because after having computed a maximal asynchronization $\Aprog'$ in a recursive call with parameter $\statement_{w}$, any smaller sound asynchronization is smaller than some predecessor in $\NextEle(\Aprog',\textcolor{blue}{\statement_{w}})$. Also, it can not output the same asynchonization twice.
Let $\Aprog^1$ and $\Aprog^2$ be two predecessors in $\NextEle(\Aprog',\textcolor{blue}{\statement_{w}})$ obtained by moving up the awaits $\statement_{w}^1$ and $\statement_{w}^2$, respectively, and assume that $\statement_{w}^1\wao \statement_{w}^2$. 
Then, all solutions computed in the recursive call on $\Aprog^1$ will have $\statement_{w}^2$ placed as in $\Aprog'$  
while all the solutions computed in the recursive call on $\Aprog^2$ will have $\statement_{w}^2$ closer to the matching call. 
Therefore, the sets of solutions computed in these two recursion branches are distinct. 

\vspace{-1mm}
\begin{theorem}\label{theorem:AsySyn}
\emph{\textsc{AsySyn}($\WAsyof{\aprog, \alib, \Alib},\statement_{w}$)}, where $\statement_{w}$ is maximal in $\WAsyof{\aprog, \alib, \Alib}$ w.r.t. $\wao$, outputs all sound asynchronizations of $\aprog[\alib]$ w.r.t. $\Alib$.
\vspace{-1mm}
\end{theorem}

The delay complexity of Algorithm~\ref{algo0} is polynomial time modulo an oracle that returns a maximal asynchronization relative to a given one. In the next section, we show that the latter problem can be reduced in polynomial time to the reachability problem in sequential programs.
\vspace{-5pt}
\section{Computing Maximal Asynchronizations}\label{sec:6}
\vspace{-5pt}

In this section, we present an implementation of the procedure \textsc{MaxRel} that relies on a reachability oracle.
In particular, we first describe an approach for computing the maximal asynchronization relative to a given asynchronization $\Aprog$, which can be seen as a way of repairing $\Aprog$ so that it becomes data-race free. Intuitively, we repeatedly eliminate data races in $\Aprog$ by moving certain \plog{await} statements closer to the matching calls. The data races in $\Aprog$ (if any) are enumerated in a certain order that prioritizes data races between actions that occur first in executions of the original synchronous program. This order allows to avoid superfluous repair steps. 

\vspace{-5pt}
\subsection{Data Race Ordering}\label{sec:DRordering}
\vspace{-3pt}

An action $\action$ representing a read/write access in a trace $\tau$ of an asynchronization $\Aprog$ of $\aprog$ is \emph{synchronously reachable} if there is an action $\action'$ in a trace $\tau'$ of $\aprog$ that represents the same statement, i.e., $\acttost(\action) = \acttost(\action')$.
It can be proved that any trace of an asynchronization contains a data race if it contains a data race between two synchronously reachable actions (see Appendix~\ref{sec:Appendix6}). In the following, we focus on data races between actions that are synchronously reachable.  

We define an order between such data races 
based on the order between actions in executions of the original synchronous program $\aprog$. 
This order relates data races in possibly different executions or asynchronizations of $\aprog$, which is possible because each action in a data race corresponds to a statement in $\aprog$. 


For two read/write statements $\statement$ and $\statement'$, $\statement\prec \statement'$ denotes the fact that there is an execution of $\aprog$ in which the \emph{first} time $\statement$ is executed occurs before the \emph{first} time $\statement'$ is executed. For two actions $\action$ and $\action'$ in an execution/trace of an asynchronization, generated by two read/write statements $\statement=\acttost(a)$ and $\statement'=\acttost(a')$, $\action\prec_\so \action'$ holds if $\statement\prec \statement'$ and either $\statement'\not\prec \statement$ or $\statement'$ is reachable from $\statement$ in the interprocedural\footnote{The interprocedural graph is the union of the control-flow graphs of each method along with edges from call sites to entry nodes, and from exit nodes to return sites.} control-flow graph of $\aprog$ without taking any back edge\footnote{A back edge points to a block that has already been met during a depth-first traversal of the control-flow graph, and corresponds to loops.}. For a \emph{deterministic} synchronous program (admitting a single execution), $\action\prec_\so \action'$ iff $\acttost(a)\prec \acttost(a')$. For non-deterministic programs, when $\acttost(a)$ and $\acttost(a')$ are contained in a loop body, it is possible that $\acttost(a)\prec \acttost(a')$ and $\acttost(a')\prec \acttost(a)$. 
In this case, we use the control-flow order to break the tie between $\action$ and $\action'$.

The order between data races corresponds to the colexicographic order induced by $\prec_\so$. This is a partial order since actions may originate from different control-flow paths and are incomparable w.r.t. $\prec_\so$. 

\vspace{-1.5mm}
\begin{definition}[Data Race Order] \label{def:dro}
  Given two races $(\action_1,\action_2)$ and $(\action_3,\action_4)$ admitted by (possibly different) asynchronizations of a synchronous program $\aprog$, we have that $(\action_1,\action_2)\prec_\so(\action_3,\action_4)$ iff $\action_2 \prec_\so \action_4$, or $\action_2 = \action_4$ and $\action_1\prec_\so \action_3$.
\vspace{-2mm}
\end{definition}

Repairing a minimal data race $(\action_1,\action_2)$ w.r.t. $\prec_\so$ removes any other 
data race $(\action_1,\action_4)$ with $(\action_2, \action_4) \in \hbo{}$ (note that we cannot have $(\action_4, \action_2) \not\in \hbo{}$ since $\action_2\prec_\so \action_4$). The repair will enforce that $(\action_1, \action_2) \in \hbo{}$ which implies that $(\action_1, \action_4) \in \hbo{}$.  

\vspace{-1mm}

\vspace{-5pt}
\subsection{Repairing Data Races}\label{ssec:repair}
\vspace{-3pt}
Repairing a data race $(\action_1, \action_2)$ reduces to modifying the position of a certain \plog{await}. 
We consider only repairs where \plog{await}s are moved up (closer to the matching call). The ``completeness'' of this set of repairs follows from the particular order in which we enumerate data races. 

\begin{wrapfigure}{r}{0.35\textwidth}
\vspace{-0.1cm}
\includegraphics[width=\linewidth]{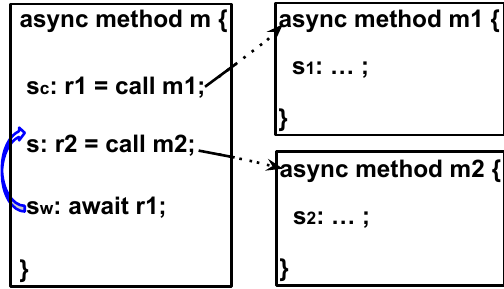}
\vspace{-0.5cm}
\caption{A data race repair.}
\label{fig:dataRacesRepair}
\vspace{-.7cm}
\end{wrapfigure}
Let $\statement_1$ and $\statement_2$ be the statements generating $\action_1$ and $\action_2$. 
In general, there exists a method $\ameth$ that (transitively) calls another asynchronous method $\ameth 1$ that contains $\statement_1$ and before awaiting for $\ameth 1$ it (transitively) calls a method $\ameth 2$ that executes $\statement_2$. This is pictured in Fig.~\ref{fig:dataRacesRepair}. It is also possible that $\ameth$ itself contains $\statement_2$ (see the program on the right of Fig.~\ref{fig:asynchronization}). 
The repair consists in moving the await for $\ameth 1$ before the call to $\ameth 2$ since this implies that $\statement_1$ will always execute before $\statement_2$ (and the corresponding actions are related by happens-before).

Formally, any two racing actions have a common ancestor in the call order $\co$ which is a call action.   
The least common ancestor of $\action_1$ and $\action_2$ in $\co$ among call actions is denoted by $\lca_\co(\action_1,\action_2)$. In Fig.~\ref{fig:dataRacesRepair}, it corresponds to the call statement $\statement_c$. More precisely, $\lca_\co(\action_1,\action_2)$ is a call action $\action_c = (\_,\aniden, \callact(\anidenp))$ s.t. $(\action_c, \action_1) \in \co$, $(\action_c, \action_2) \in \co$, and for each other call action $\action_c'$, if $(\action_c,\action_c')\in \co$ then $(\action_c', \action_1) \not\in \co$. 
This call action represents an asynchronous call for which the matching await $\statement_w$ must move to repair the data race. The await should be moved before the last statement in the same method generating an action which precedes $\action_2$ in the reflexive closure of call order (statement $\statement$ in Fig.~\ref{fig:dataRacesRepair}). 
This way every statement that follows $\statement_c$ in call order will be executed before $\statement$ and before any statement which succeeds $\statement$ in call order, including $\statement_2$. Note that moving the await $\statement_w$ anywhere after $\statement$ will not affect the concurrency between $\action_1$ and $\action_2$. 

The pair $(\statement_c,\statement)$ 
is called the \emph{root cause} of the data race $(a_1,a_2)$. 
Let $\repairdatarace(\Aprog, \statement_c,\statement)$ be the maximal asynchronization $\Aprog'$ smaller than $\Aprog$ w.r.t. $\leq$, s.t. no await statement matching $\statement_c$ occurs after $\statement$ on a CFG path. 

\vspace{-5pt}
\subsection{A Procedure for Computing Maximal Asynchronizations}
\vspace{-3pt}

Given an asynchronization $\Aprog$, the procedure \textsc{MaxRel} in Algorithm~\ref{algo2} computes the maximal asynchronization relative to $\Aprog$ by repairing data races iteratively until the program becomes data race free. 
The sub-procedure $\textsc{RCMinDR}(\Aprog')$ computes the root cause of a minimal data race $(\action_1, \action_2)$ of $\Aprog'$ w.r.t. $\prec_\so$ such that the two actions are synchronously reachable. If $\Aprog'$ is data race free, $\textsc{RCMinDR}(\Aprog')$ returns $\bot$. The following theorem states the correctness of \textsc{MaxRel}.

\setlength{\textfloatsep}{7pt}
\begin{algorithm}[t]
  \caption{The procedure \textsc{MaxRel} to find the maximal asynchronization of $\aprog$ relative to $\Aprog$. }\label{algo2}
  \begin{algorithmic}[1]
  \Procedure{MaxRel}{$\Aprog$}
  \State $\ \ \Aprog' \leftarrow \Aprog$
  \State $\ \ \mathit{root} \leftarrow \textsc{RCMinDR}(\Aprog')$
  \State $\ \ \textbf{while}\ \mathit{root} \neq \bot$
  \State $\ \ \ \ \ \ \ \Aprog'  \leftarrow \repairdatarace(\Aprog', \mathit{root})$
  \State $\ \ \ \ \ \ \ \mathit{root} \leftarrow \textsc{RCMinDR}(\Aprog')$
  \State $\ \ \textbf{return}\ \Aprog'$
  \EndProcedure
  \end{algorithmic}
\end{algorithm}

\vspace{-1.5mm}
\begin{theorem} \label{theorem:anc3}
  Given an asynchronization $\Aprog\in \Asyof{\aprog, \alib, \Alib}$, \textsc{MaxRel}($\Aprog$) returns the maximal asynchronization of $\aprog$ relative to $\Aprog$.
\vspace{-1.5mm}
\end{theorem}

\textsc{MaxRel}($\Aprog$) repairs a number of data races which is linear in the size of the input. Indeed, each repair results in moving an await closer to the matching call and before at least one more statement from the original program $\aprog$.


The problem of computing root causes of minimal data races 
is reducible to reachability (assertion checking) in sequential programs. This reduction builds on a program instrumentation for checking if there exists a data race that involves two given statements $(\statement_1,\statement_2)$ that are reachable in an executions of $P$.
This instrumentation is used in an iterative process where pairs of statements are enumerated according to the colexicographic order induced by $\prec$. 
For lack of space, we present only the main ideas of the instrumentation (see Appendix~\ref{sec:recApp}). The instrumentation simulates executions of an asynchronization $\Aprog$ using non-deterministic synchronous code where methods may be only partially executed (modeling \plog{await} interruptions). Immediately after executing $\statement_1$,
the current invocation $t_1$ is interrupted (by executing a \plog{return} added by the instrumentation). The active invocations that transitively called $t_1$ are also interrupted when reaching an \plog{await} for an invocation in this call chain (the other invocations are executed until completion as in the synchronous semantics). When reaching $\statement_2$, if $\statement_1$ has already been executed and at least one invocation has been interrupted, which means that $\statement_1$ is concurrent with $\statement_2$, then the instrumentation stops with an assertion violation. The instrumentation also computes the root cause of the data race using additional variables for tracking call dependencies.

\vspace{-5pt}
\section{Asymptotic Complexity of Asynchronization Synthesis} \label{sec:asymComplexity}
\vspace{-5pt}

We state the complexity of the asynchronization synthesis problem. Algorithm~\ref{algo0} shows that the delay complexity of this problem is polynomial-time in the number of statements in input program modulo the complexity of computing a maximal asynchronization, which Algorithm~\ref{algo2} shows to be polynomial-time reducible to reachability in sequential programs.
Since the reachability problem is PSPACE-complete for finite-state sequential programs~\cite{DBLP:conf/tacas/GodefroidY13}, we get the following:

\vspace{-1mm}
\begin{theorem}\label{th:compl1}
The output complexity\footnote{Note that all asynchronizations can be enumerated with polynomial space.} and delay complexity of the asynchronization synthesis problem is polynomial time modulo an oracle for reachability in sequential programs, and PSPACE for finite-state programs.
\vspace{-1mm}
\end{theorem}

This result is optimal, i.e., checking whether there exists a sound asynchronization which is different from the trivial strong synchronization is PSPACE-hard (follows from a reduction from the reachability problem). See Appendices~\ref{sec:recApp} and \ref{appendixF} for the detailed formal proofs. 

\vspace{-5pt}
\section{Asynchronization Synthesis Using Data-Flow Analysis}\label{sec:analysis}
\vspace{-5pt}
In this section, we present a refinement of Algorithm~\ref{algo2} that relies on a bottom-up inter-procedural data flow analysis. 
The analysis is used to compute maximal asynchronizations for abstractions of programs where every Boolean condition (in if-then-else or while statements) is replaced with the non-deterministic choice $*$, and used as an implementation of \textsc{MaxRel} in Algorithm~\ref{algo0}.

For a program $\aprog$, we define an abstraction $\Cprog$ where every conditional \texttt{if} $\langle le\rangle$  $\{S_1\}$ \texttt{else} $\{S_2\}$ is rewritten to \texttt{if} $*$ $\{S_1\}$ \texttt{else} $\{S_2\}$, and every \texttt{while} $\langle le\rangle$ $\{S\}$ is rewritten to \texttt{if} $*$ $\{S\}$. Besides adding the non-deterministic choice $*$, loops are unrolled exactly once. Every asynchronization $\Aprog$ of $\aprog$ corresponds to an abstraction $\CAprog$ obtained by applying exactly the same rewriting.
$\Cprog$ is a sound abstraction of $\aprog$ in terms of sound asynchronizations it admits. Unrolling loops once is sound because every asynchronous call in a loop iteration should be awaited for in the same iteration (see the syntactic constraints in Section \ref{sec:1}).

\vspace{-1mm}
\begin{theorem}\label{lem:absSound}
  If $\CAprog$ is a sound asynchronization of $\Cprog$ w.r.t. $\Alib$, then $\Aprog$ is a sound asynchronization of $\aprog$ w.r.t. $\Alib$.
\vspace{-1mm}
\end{theorem}

The procedure for computing maximal asynchronizations of $\Cprog$ relative to a given asynchronization $\CAprog$ traverses methods of $\CAprog$ in a bottom-up fashion, detects data races using summaries of read/write accesses computed using a straightforward data-flow analysis, and repairs data races using the schema presented in Section~\ref{ssec:repair}. Applying this procedure to a real programming language requires an alias analysis to detect statements that may access the same memory location (this is trivial in our language which is used to simplify the exposition).

We consider an enumeration of methods called \emph{bottom-up order}, which is the reverse of a topological ordering of the call graph\footnote{The nodes of the call graph are methods and there is an edge from a method $\ameth_1$ to a method $\ameth_2$ if $\ameth_1$ contains a call statement that calls $\ameth_2$.}. 
For each method $\ameth$, let $\mathcal{R}(\ameth)$ be the set of program variables that $\ameth$ can read, which is defined as the union of $\mathcal{R}(\ameth')$ for every method $\ameth'$ called by $\ameth$ and the set of program variables read in statements in the body of $\ameth$. The set of variables $\mathcal{W}(\ameth)$ that $\ameth$ can write is defined in a similar manner. We define $\mathsf{RW\text{-}var}(\ameth)=(\mathcal{R}(\ameth),\mathcal{W}(\ameth))$. We extend the notation $\mathsf{RW\text{-}var}$ to statements as follows: $\mathsf{RW\text{-}var}(\langle r\rangle := \langle x\rangle)=(\{x\},\emptyset)$, $\mathsf{RW\text{-}var}(\langle x\rangle := \langle le\rangle)=(\emptyset,\{x\})$, $\mathsf{RW\text{-}var}(r := \mathtt{call}\ \ameth)=\mathsf{RW\text{-}var}(\ameth)$, and $\mathsf{RW\text{-}var}(\statement)=(\emptyset,\emptyset)$, for any other type of statement $s$. Also, let $\mathsf{CRW\text{-}var}(\ameth)$ be the set of read or write accesses that $\ameth$ can do and that can be concurrent with accesses that a caller of $\ameth$ can do after calling $\ameth$. These correspond to read/write statements that follow an \plog{await} in $\ameth$, or to accesses in $\mathsf{CRW\text{-}var}(\ameth')$ for a method $\ameth'$ called by $\ameth$. These sets of accesses can be computed using the following data-flow analysis: for all methods $\ameth\in\CAprog$ in bottom-up order, and for each statement $\statement$ in the body of $\ameth$ from begin to end,
\begin{itemize}[noitemsep,topsep=1pt]
			\item if $\statement$ is a call to $\ameth'$ and $\statement$ is \emph{not} reachable from an \plog{await} in the CFG of $\ameth$
			\vspace{.5mm}
			\begin{itemize}[noitemsep,topsep=1pt]
				\item[$\bullet$] $\mathsf{CRW\text{-}var}(\ameth) \leftarrow  \mathsf{CRW\text{-}var}(\ameth) \cup \mathsf{CRW\text{-}var}(\ameth')$
			\end{itemize}
			\item if $\statement$ is reachable from an \plog{await} statement in the CFG of $\ameth$
			\vspace{.5mm}
			\begin{itemize}[noitemsep,topsep=1pt]
				\item[$\bullet$] $\mathsf{CRW\text{-}var}(\ameth)  \leftarrow  \mathsf{CRW\text{-}var}(\ameth) \cup \mathsf{RW\text{-}var}(\statement)$
			\end{itemize}			
\end{itemize}
\vspace{0.5mm}
We use $(\mathcal{R}_1,\mathcal{W}_1)\bowtie(\mathcal{R}_2,\mathcal{W}_2)$ to denote the fact that $\mathcal{W}_1\cap (\mathcal{R}_2\cup \mathcal{W}_2)\neq\emptyset$ or $\mathcal{W}_2\cap (\mathcal{R}_1\cup \mathcal{W}_1)\neq\emptyset$ (i.e., a conflict between read/write accesses). We define the procedure \textsc{MaxRel}$^\#$ that given an asynchronization $\CAprog$ works as follows: 
\begin{itemize}[topsep=5pt] 
	\item for all methods $\ameth\in\CAprog$ in bottom-up order, and for each statement $\statement$ in the body of $\ameth$ from begin to end,
		\begin{itemize}[noitemsep,topsep=0pt]
			\item if $\statement$ occurs between $r := \mathtt{call}\ \ameth'$ and $\mathtt{await}\ r$ (for some $\ameth'$), and $\mathsf{RW\text{-}var}(\statement) \bowtie \mathsf{CRW\text{-}var}(\ameth')$, then $\CAprog\leftarrow \repairdatarace(\CAprog,r := \mathtt{call}\ \ameth',s)$
		\end{itemize}
	\item return $\CAprog$
\end{itemize}
\vspace{-2mm}
\begin{theorem}
\mbox{\textsc{MaxRel}$^\#(\CAprog)$ returns a maximal asynchronization relative to $\CAprog$.}
\vspace{-6mm}
\end{theorem}
Since \textsc{MaxRel}$^\#$ is based on a single bottom-up traversal of the call graph of the input asynchronization $\CAprog$ we get the following result.
\vspace{-1mm}
\begin{theorem}
The delay complexity of the asynchronization synthesis problem restricted to abstracted programs $\Cprog$ is polynomial time.
\vspace{-3mm}
\end{theorem}

\vspace{-6pt}
\section{Experimental Evaluation}
\label{sec:experiments}
\vspace{-5pt}

We present an empirical evaluation of our asynchronization synthesis approach, where maximal asynchronizations are computed using the data-flow analysis in Section~\ref{sec:analysis}. Our benchmark consists mostly of asynchronous C\# programs from open-source GitHub projects. 
We evaluate the effectiveness in reproducing the original program as an asynchronization of a program where asynchronous calls are reverted to synchronous calls, along with other sound asynchronizations.
    
\noindent
\textbf{Implementation.} 
We developed a prototype tool 
that uses the Roslyn .NET compiler platform~\cite{roslyn} to construct CFGs for methods in a C\# program. This prototype supports C\# programs written in static single assignment (SSA) form that include basic conditional/looping constructs and async/await as concurrency primitives. 
Note that object fields are interpreted as program variables in the terminology of \S\ref{sec:1} (data races concern accesses to object fields). 
It assumes that alias information is provided apriori; these constraints can be removed in the future with more engineering effort. In general, our synthesis procedure is compatible with any sound alias analysis. The precision of this analysis impacts only the set (number) of asynchronizations outputted by the procedure (a more precise analysis may lead to more sound asynchronizations).

The tool takes as input a possibly asynchronous program, and a mapping between synchronous and asynchronous variations of base methods in this program. It reverts every asynchronous call to a synchronous call, and it enumerates sound asynchronizations of the obtained program (using Algorithm~\ref{algo0}). 
\noindent
\textbf{Benchmark.}
Our evaluation uses a benchmark listed in Table~\ref{tab:experiments}, which
contains $5$ synthetic examples (variations of the program in Fig.~\ref{fig:example}), $9$ programs extracted from open-source C\# GitHub projects (their name is a prefix of the repository name), and $2$ programs inspired by questions on \texttt{stack\-overflow.com} about async\-/await in C\# (their name ends in Stack\-overflow). 
%
Overall, there are 13 base methods involved in computing asynchronizations of these programs (having both synchronous and asynchronous versions), coming from $5$ C\# libraries  (\emph{System.\-IO}, \emph{System.\-Net}, \emph{Windows.\-Storage}, \emph{Microsoft.\-Windows\-Azure.\-Storage}, and \emph{Microsoft.\-Azure.\-Devices}). They are modeled as described in Section~\ref{sec:1}.




\begin{table}[t]
            \caption{Empirical results. Syntactic characteristics of input programs: lines of code (loc), number of methods (m), number of method calls (c), number of asynchronous calls (ac), number of awaits that \emph{could} be placed at least one statement away from the matching call (await$_{\#}$). Data concerning the enumeration of asynchronizations: number of awaits that \emph{were} placed at least one statement away from the matching call (await), number of races discovered and repaired (races), number of statements that the awaits in the maximal asynchronization are covering \emph{more than} in the input program (cover), number of computed asynchronizations (async), and running time (t).}
            \label{tab:experiments}
            \vspace{-0mm}
                \scriptsize
                \setlength{\tabcolsep}{0.7em}
\begin{tabular}{lllllllllll}
\toprule
Program                     & loc & m  & c & ac & await${_\#}$ & await & races & cover & async & t(s)  \\ \cmidrule(lr){1-11}
SyntheticBenchmark-1        & 77  & 3  & 6   & 5      & 4   & 4  & 5  & 0     & 9     & 1.4     \\ 
SyntheticBenchmark-2        & 115 & 4  & 12  & 10     & 6   & 3  & 3  & 0     & 8     & 1.4     \\ 
SyntheticBenchmark-3        & 168 & 6  & 16  & 13     & 9   & 7  & 4  & 0     & 128   & 1.5     \\
SyntheticBenchmark-4        & 171 & 6  & 17  & 14     & 10  & 8  & 5  & 0     & 256   & 1.9    \\
SyntheticBenchmark-5        & 170 & 6  & 17  & 14     & 10  & 8  & 9  & 0     & 272   & 2   \\
Azure-Remote                & 520 & 10 & 14  & 5      & 0   & 0  & 0  & 0     & 1     & 2.2   \\
Azure-Webjobs               & 190 & 6  & 14  & 6      & 1   & 1  & 0  & 1     & 3     & 1.6     \\
FritzDectCore               & 141 & 7  & 11  & 8      & 1   & 1  & 0  & 1     & 2     & 1.6     \\
MultiPlatform               & 53  & 2  & 6   & 4      & 2   & 2  & 0  & 2     & 4     & 1.1     \\
NetRpc                      & 887 & 13 & 18  & 11     & 4   & 1  & 3  & 0     & 3     & 2     \\
TestAZureBoards             & 43  & 3  & 3   & 3      & 0   & 0  & 0  & 0     & 1     & 1.5     \\
VBForums-Viewer             & 275 & 7  & 10  & 7      & 3   & 2  & 1  & 1     & 6     & 1.8     \\
Voat                        & 178 & 3  & 5   & 5      & 2   & 1  & 1  & 1     & 3     & 1.2    \\
WordpressRESTClient         & 133 & 3  & 10  & 8      & 4   & 2  & 1  & 0     & 4     & 1.7     \\
ReadFile-Stackoverflow      & 47  & 2  & 3   & 3      & 1   & 0  & 1  & 0     & 1     & 1.5     \\
UI-Stackoverflow            & 50  & 3  & 4   & 4      & 3   & 3  & 3  & 0     & 12    & 1.5 \\    
\bottomrule
\end{tabular}

\end{table}

\noindent
\textbf{Evaluation.} 
The last five columns of Table~\ref{tab:experiments} list data concerning the application of our tool. The column async lists the number of outputted sound asynchronizations. In general, the number of asynchronizations depends on the number of invocations (column ac) and the size of the code blocks between an invocation and the instruction using its return value (column await$_{\#}$ gives the number of non-empty blocks).  
The number of \emph{sound} asynchronizations depends roughly, on how many of these code blocks are racing with the method body. These asynchronizations contain \plog{await}s that are at a non-zero distance from the matching call (non-zero values in column await) and for many Github programs, this distance is bigger than in the original program (non-zero values in column cover).
This shows that we are able to increase the distances between \plog{await}s and their matching calls for those programs. The distance between \plog{await}s and matching calls in maximal asynchronizations of non synthetic benchmarks is $1.27$ statements on average. A statement representing a method call is counted as one independently of the method's body size. With a single level of inlining, the number of statements becomes 2.82 on average. However, these statements are again, mostly IO calls (access to network or disk) or library calls (string/bytes formatting methods) whose 
execution time is not negligible.
The running times for the last three synthetic benchmarks show that our procedure is scalable when programs have a large number of sound asynchronizations.

With few exceptions, each program admits multiple sound asynchronizations (values in column async bigger than one), which makes the focus on the delay complexity relevant. This leaves the possibility of making a choice based on other criteria, e.g., performance metrics. As shown by the examples in Fig.~\ref{fig:performance}, 
their performance can be derived only dynamically (by executing them). 
These results show that our techniques have the potential of becoming the basis of a refactoring tool allowing programmers to improve their usage of the async/await primitives. The artifacts are available in a GitHub repository \cite{artifact7055422}.

\vspace{-5pt}
\section{Related Work}
\label{sec:related}
\vspace{-5pt}
There are many works on synthesizing or repairing concurrent programs in the standard multi-threading model, e.g., automatic parallelization in compilers~\cite{DBLP:journals/csur/BaconGS94,DBLP:journals/computer/BlumeDEGHLLPPPRT96,DBLP:conf/ipps/HanT01}, or  synchronization synthesis~\cite{DBLP:conf/cav/CernyHRRT13,DBLP:conf/cav/CernyHRRT14,DBLP:conf/spin/ClarkeE08,DBLP:journals/toplas/MannaW84,DBLP:conf/popl/VechevYY10,DBLP:conf/tacas/VechevYY09,DBLP:conf/fmcad/BloemHKKAS14,DBLP:conf/cav/CernyCHRRST15,DBLP:conf/popl/GuptaHRST15}. We focus on the use of async/await which poses specific challenges that are not covered in these works. 

Our semantics without $\mathtt{await}\ *$ instructions is equivalent to the semantics defined in \cite{DBLP:conf/ecoop/BiermanRMMT12,DBLP:conf/pldi/SanthiarK17}. But, to simplify the exposition, we consider a more restricted programming language. For the modeling of asynchronous IO operations, we follow \cite{DBLP:conf/ecoop/BiermanRMMT12} with the restriction that the code following an $\mathtt{await}\ *$ is executed atomically. This is sound when focusing on data-race freedom because even if executed atomically, any two instructions from different asynchronous IO operations (following $\mathtt{await}\ *$) are not happens-before related.

\noindent
\textbf{Program Refactoring.} 
Program refactoring tools have been proposed for converting C\# programs using explicit callbacks into async/await programs~\cite{DBLP:conf/icse/OkurHDD14} or Android programs using AsyncTask into programs that use IntentService~\cite{DBLP:conf/kbse/LinOD15}. 
The C\# tool~\cite{DBLP:conf/icse/OkurHDD14}, which is the closest to our work, makes it possible to repair misusage of async/await that might result in deadlocks. 
This tool cannot modify procedure calls to be asynchronous as in our work. 
A static analysis based technique for refactoring JavaScript programs is proposed in~\cite{DBLP:journals/pacmpl/GokhaleTT21}. 
As opposed to our work, this refactoring technique is unsound in general. It requires that programmers review the refactoring for correctness, which is error-prone. Also, in comparison to~\cite{DBLP:journals/pacmpl/GokhaleTT21}, we carry a formal study of the more general problem of finding all sound asynchronizations and investigate its complexity. 

\noindent
\textbf{Data Race Detection.}
Many works study dynamic data race detection using happens-before and lock-set analysis, or timing-based detection~\cite{DBLP:conf/sosp/LiLMNP19,DBLP:conf/pldi/KiniM017,DBLP:conf/popl/SmaragdakisESYF12,DBLP:conf/rv/RamanZSVY10,DBLP:conf/pldi/FlanaganF09}. 
They could be used to approximate our reduction from data race checking to reachability in sequential programs.
Some works~\cite{DBLP:journals/pacmpl/BlackshearGOS18,DBLP:conf/pldi/LiuH18,DBLP:conf/sosp/EnglerA03} propose static analyses for finding data races. 
\cite{DBLP:journals/pacmpl/BlackshearGOS18} designs a compositional data race detector for multi-threaded Java programs, based on an inter-procedural analysis assuming that any two public methods can execute in parallel. Similar to~\cite{DBLP:conf/pldi/SanthiarK17}, they precompute method summaries to extract potential racy accesses. 
These approaches are similar to the analysis in Section~\ref{sec:analysis}, but they concern a different programming model.

\noindent
\textbf{Analyzing Asynchronous Programs.} 
Several works propose program analyses for various classes of asynchronous programs. 
\cite{DBLP:conf/popl/BouajjaniE12,DBLP:journals/toplas/GantyM12} give complexity results for the reachability problem, and \cite{DBLP:conf/pldi/SanthiarK17} proposes a static analysis for deadlock detection in C\# programs that use both asynchronous and synchronous wait primitives. 
\cite{DBLP:conf/esop/BouajjaniEEOT17} investigates the problem of checking whether Java UI asynchronous programs have the same set of behaviors as sequential programs where roughly, asynchronous tasks are executed synchronously. 
\vspace{-5pt}
\section{Conclusion}\label{sec:conc}
\vspace{-5pt}
We proposed a framework for refactoring sequential programs to equivalent asynchronous programs based on async/await. We determined precise complexity bounds for the problem of computing 
all sound asynchronizations. This problem makes it possible to compute a sound asynchronization that maximizes performance by separating concerns -- enumerate sound asynchronizations and evaluate performance separately.
On the practical side, we have introduced an approximated synthesis procedure based on data-flow analysis that we implemented and evaluated on a benchmark of non-trivial C\# programs.

%
%
%
%
%
The asynchronous programs rely exclusively on async/await and are deadlock-free by definition. Deadlocks can occur in a mix of async/await with ``explicit'' multi-threading that includes blocking \texttt{wait} primitives. Extending our approach for such programs is an interesting direction for future work. 

    \bibliographystyle{splncs04}
    \bibliography{draft}

    \appendix
    \newpage
\section{Formalization and Proofs of Section~\ref{sec:4}}

The following lemma shows that the absence of data races implies equivalence to the original program, in the sense of reaching the same set of configurations (program variable valuations).

\vspace{-1mm}
\begin{lemma}\label{lemma:soundDR}
  $\Aprog[\Alib]$ is sound implies $\pstatesconf[\aprog[\alib]] = \pstatesconf[\Aprog[\Alib]]$, for every $\Aprog[\Alib]\in \Asyof{\aprog, \alib, \Alib}$
\vspace{-1mm}
\end{lemma}

\begin{proof}[Proof of Lemma~\ref{lemma:soundDR}]
  Let $\rho$ be an execution of $\Aprog$ that reaches a configuration $\psconf \in \pstatesconf[\Aprog]$. We show that actions in  $\rho$ can be reordered such that any action that occurs in $\rho$ between $(\_,\aniden, \callact(\anidenp))$ and $(\_,\anidenp, \returnact)$ is not of the form $(\_,\aniden, \_)$ (i.e., the task $\anidenp$ is executed synchronously). If an action $(\_,\aniden, \_)$ occurs in $\rho$ between $(\_,\aniden, \callact(\anidenp))$ and $(\_,\anidenp, \returnact)$, then it must be concurrent with $(\anidenp, \returnact)$.  Since $\Aprog$ does not admit data races, an execution $\rho'$ resulting from $\rho$ by reordering any two concurrent actions reaches the same configuration $\psconf$ as $\rho$. Therefore, there exists an execution $\rho''$ where the actions that occur between any $(\_,\aniden, \callact(\anidenp))$ and $(\_,\anidenp, \returnact)$ are not of the form $(\_,\aniden, \_)$. This is also an execution of $\aprog$ (modulo removing the awaits which have no effect), which implies $\psconf \in \pstatesconf[\aprog]$. 
\end{proof}

\section{Formalization and Proofs of Section~\ref{sec:synthe-corrctProg}}\label{sec:synthe-corrctProg-app}

The following lemma shows that for a given $\Aprog$ there exists a unique $\Aprog'$ that is a maximal asynchronization of $\aprog$ relative to $\Aprog$. The existence is implied by the fact that $\BAsyof{\aprog, \alib, \Alib}$ is the bottom element of $\leq$. To prove uniqueness, we assume by contradiction that there exist two incomparable maximal asynchronizations $\Aprog^1$ and $\Aprog^2$ and select the first await statement $\statement_{w}^{1}$, according to the control-flow of the sequential program, that is placed in different positions in the two programs. Assume that $\statement_{w}^{1}$ is closer to its matching call in $\Aprog^1$. Then, we move $\statement_{w}^{1}$ in $\Aprog^1$ further away from its matching call to the same position as in $\Aprog^2$. This modification does not introduce data races since $\Aprog^2$ is data race free. Thus, the resulting program is data race free, bigger than $\Aprog^1$, and smaller than $\Aprog$ w.r.t. $\leq$ contradicting the fact that $\Aprog^1$ is a maximal asynchronization.

\begin{lemma}\label{lemma:optimalAsync}
  Given an asynchronization $\Aprog \in \Asyof{\aprog, \alib, \Alib}$,
  there exists a unique program $\Aprog'$ that is a maximal asynchronization of $\aprog$ relative to $\Aprog$.
\end{lemma}

\begin{proof}[Proof of Lemma \ref{lemma:optimalAsync}]
  Since $\BAsyof{\aprog, \alib, \Alib}$ is the bottom element of $\leq$, then there always exists a sound asynchronization smaller than $\Aprog$. Assume by contradiction that there exist two distinct programs $\Aprog^1$ and $\Aprog^2$ that are both maximal asynchronizations of $\aprog$ relative to $\Aprog$. Let $\rho^1$ (resp., $\rho^2$) be an execution of $\Aprog^1$ (resp., $\Aprog^2$) where every $\awaitact\ *$ does not suspend the execution of the current task, i.e., $\rho^1$ and $\rho^2$ simulate the synchronous execution of $\aprog$. Let $\statement_{w}^{1}$ be the statement corresponding to the first \plog{await} action in $\rho^1$ such that (1) there exists an \plog{await} action in $\rho^2$ with the corresponding \plog{await} statement $\statement_{w}^{2}$, such that $\statement_{w}^{1}$ and $\statement_{w}^{2}$ match the same call in $\aprog$, and $\cover(\statement_{w}^{1})\subset \cover(\statement_{w}^{2})$ (this holds because $\Aprog^1$ and $\Aprog^2$ are distinct asynchronizations of the same synchronous program, thus $\cover(\statement_{w}^{1})$ and $\cover(\statement_{w}^{2})$ must be comparable), and (2) for every other \plog{await} statement $\statement_{w}^{3}$ in $\Aprog^1$ that generates an \plog{await} action which occurs before the \plog{await} action  of $\statement_{w}^{1}$ in $\rho^1$, there exists an \plog{await} statement $\statement_{w}^{4}$ in $\Aprog^2$ matching the same call in $\aprog$, such that $\cover(\statement_{w}^{3}) = \cover(\statement_{w}^{4})$. 
    
  Let $\Aprog^3$ be the program obtained from $\Aprog^1$ by moving the await $\statement_{w}^{1}$ down (further away from the matching call) such that $\cover(\statement_{w}^{1}) = \cover(\statement_{w}^{2})$. 
  Moving an await down can only create data races between actions that occur after the execution of the matching call. Then, $\Aprog^3$ contains a data race iff there exists an execution $\rho$ of $\Aprog^3$ and two concurrent actions $\action_1$ and  $\action_2$ that occur between the action $(\_,i, \awaitact(j))$ generated by $\statement_{w}^{1}$ and the action $(\_,i, \callact(j))$ of the call matching $\statement_{w}^{1}$, such that:  
  $$((\_,i, \callact(j)), \action_1) \in \co{},\  
    (\action_1, \action_w) \not\in \hbo{},\ ((\_,i, \callact(j)), \action_2)  \in \co{} \mbox{\ and\ } (\action_2, (\_,i, \awaitact(j))) \in \hbo{}$$ 
  where the action $\action_w$ corresponds to the first await action in the task $j$. Let $\statement_{w}$ be the statement corresponding to the action $\action_w$. Since the only difference between $\Aprog^3$ and $\Aprog^2$ is the placement of awaits then $((\_,i, \callact(j)), \action_1) \in \co{}$ and $((\_,i, \callact(j)), \action_2) \in \co{}$ hold in any execution $\rho'$ of $\Aprog^2$ that contains the actions $\action_1$ and $\action_2$. Also, note that since $\action_w$ occurs in the task $j$ that the action of $\statement_{w}^{1}$ is waiting for. This implies that in $\rho^1$ the action of $\statement_w$ occurs before the action of $\statement_{w}^{1}$ in $\rho^1$. Therefore, by the definition of $\statement_{w}^{1}$ we have that $\statement_w$ in $\Aprog^1$ covers the same set of statements as the corresponding $\statement'_w$ in $\Aprog^2$ that matches the same call as $\statement_w$. Consequently, $(\action_1, \action'_w) \not\in \hbo{}$ and $(\action_2, (\_,i, \awaitact(j))) \in \hbo{}$ hold in any execution $\rho'$ of $\Aprog^2$ that contains the actions $\action_1$ and $\action_2$ ($\action'_w$ is the action of $\statement'_w$). Thus, there exists an execution $\rho'$ of $\Aprog^2$ such that the actions $\action_1$ and $\action_2$ are concurrent. This implies that if $\Aprog^3$ admits a data race, then $\Aprog^2$ admits a data race between actions generated by the same statements. As $\Aprog^2$ is data race free, we get that $\Aprog^3$ is data race free as well. Since $\Aprog^1 < \Aprog^3$, we get that $\Aprog^1$ is not maximal, which contradicts the hypothesis.
\end{proof}

The complexity analysis also relies on a property of the maximal asynchronization relative to an immediate predecessor: if the predecessor is defined by moving an await $\statement_{w}''$, then the maximal asynchronization is obtained by moving only awaits smaller than $\statement_{w}''$ w.r.t. $\wao$.

\begin{lemma}\label{lemma:MonotoneAwaitRepair}
  If $\Aprog''$ is an immediate predecessor of a sound asynchronization $\Aprog'$, which is defined by moving an await $\statement_{w}''$ in $\Aprog'$ up, then the maximal sound asynchronization relative to $\Aprog''$ is obtained by moving only awaits smaller than $\statement_{w}''$ w.r.t. $\wao$. 
\vspace{-1mm}
\end{lemma}

\begin{proof}[Proof of Lemma \ref{lemma:MonotoneAwaitRepair}]
  Moving an await up in $\Aprog'$ can only create data races between actions that occur after the execution of this await (because the invocation is suspended earlier). The only possible repairs of these data races consists in either moving $\statement_{w}''$ down which results in $\Aprog'$ or moving up some other awaits that occur in methods that (indirectly) call the method in which $\statement_{w}''$ occurs. The first case is not applicable because it gives a program that is not smaller than $\Aprog''$. In the second case, 
  every await $\statement_{w}'$ that is moved up occurs in a method that (indirectly) calls the method in which $\statement_{w}''$ occurs, and therefore, $\statement_{w}'$ is smaller than $\statement_{w}''$ w.r.t. $\wao$. 
\end{proof}

Before giving the proof of Theorem \ref{theorem:AsySyn}, we note that the total order relation $\wao$ between awaits is fixed throughout the recursion of AsySyn and it corresponds to the order of the awaits in the weakest asynchronization of $\aprog$, i.e., $\WAsyof{\aprog, \alib, \Alib}$. This is because the order between awaits in the same method might change from one asynchronization to another in $\Asyof{\aprog, \alib, \Alib}$. If the control-flow graph of a method contains branches, it is possible to replace all \plog{await} statements matching $\statement_c$ that are reachable in the CFG from $\statement$ with a single \plog{await} statement $\statement_{w}$, in this case $\statement_{w}$ is ordered before any other await that one of the awaits that $\statement_{w}$ replaces is ordered before and is ordered after any await that all the awaits that $\statement_{w}$ replaces are ordered before. Also, it is possible to add additional \plog{await}s statements in branches, in this case derive a total order between these awaits and order the awaits before or after any other await that the original await was ordered before or after, respectively.

\begin{proof}[Proof of Theorem \ref{theorem:AsySyn}]
Let $\Aprog$ be the weakest asynchronization of $\aprog$, then the set of all sound asynchronizations of $\aprog$ is 
$\mathcal{A} = \{\Aprog'': \pstatesconf[\Aprog''] = \pstatesconf[\aprog] \mbox{ and }\Aprog'' \leq \Aprog' \}$, 
where $\Aprog'$ is the maximal asynchronization of $\aprog$ relative to $\Aprog$. 
It is clear that every asynchronization outputted by \emph{\textsc{AsySyn}($\Aprog,\statement_{w}^{0}$)} is in the set $\mathcal{A}$. 

Let $\Aprog^0$ be a sound asynchronization of $\aprog[\alib]$ w.r.t. $\Alib$, i.e., $\Aprog^0 \in \mathcal{A}$. We will show that \emph{\textsc{AsySyn}} outputs $\Aprog^0$. We have that either $\Aprog^0 = \Aprog'$ or $\Aprog^0 < \Aprog'$. 
The first case implies that $\Aprog^0$ is in \emph{\textsc{AsySyn}($\Aprog,\statement_{w}$)}. For the second case: let $\statement_{w}^{1}$ be the maximum element in $\Aprog'$ w.r.t. $\wao$ that matches the same call as $\statement_{w}^{1'}$ in $\Aprog^0$ s.t. $\cover(\statement_{w}^{1'}) \subset \cover(\statement_{w}^{1})$. Then, let $(\Aprog^1,\statement_{w}^{1}) \in \NextEle(\Aprog',\statement_{w})$. We obtain that either $\Aprog^0 = \Aprog^1$ or $\Aprog^0 < \Aprog^1$. The fist case implies that $\Aprog^0$ is in \emph{\textsc{AsySyn}($\Aprog,\statement_{w}$)}. 
For the second case: let $\Aprog^{1'} = \textsc{MaxRel}(\Aprog^1)$ then either $\Aprog^0 = \Aprog^{1'}$ or $\Aprog^0 < \Aprog^{1'}$. 
The fist case implies that $\Aprog^0$ is in \emph{\textsc{AsySyn}($\Aprog,\statement_{w}$)}. For the second case: let $\statement_{w}^{2}$ be the maximum element in $\Aprog^{1'}$ w.r.t. $\wao$ that matches the same call as $\statement_{w}^{2'}$ in $\Aprog^0$ s.t. $\cover(\statement_{w}^{2'}) \subset \cover(\statement_{w}^{2})$. 
Since $\Aprog^1$ is an immediate successor of $\Aprog'$ by moving the await $\statement_{w}^{1}$, then Lemma~\ref{lemma:MonotoneAwaitRepair} implies $\Aprog^{1'}$ is obtained by moving only awaits smaller than $\statement_{w}^{1}$ w.r.t. $\wao$. 
Then, we either have $\statement_{w}^{2}= \statement_{w}^{1}$ or $\statement_{w}^{2}\ \wao\ \statement_{w}^{1}$. Thus,  $(\Aprog^2,\statement_{w}^{2}) \in \NextEle(\Aprog^{1'},\statement_{w}^{1})$. We then obtain that either $\Aprog^0 = \Aprog^2$ or $\Aprog^0 < \Aprog^2$. Then, we repeat the above proof process until we obtain $\Aprog^n = \Aprog^0$. Thus, \emph{\textsc{Asy\-Syn}} outputs $\Aprog^0$. 

Let $\statement_{w}^{1}$ and $\statement_{w}^{2}$ be two distinct await statements in $\Aprog'$ s.t. $\statement_{w}^{2}\ \wao\ \statement_{w}^{1}$ and  $(\Aprog^1,\statement_{w}^{1}),\ (\Aprog^2,\statement_{w}^{2})  \in \NextEle(\Aprog',\statement_{w})$. Similar to before then we have that $\Aprog^{2'} = \textsc{MaxRel}(\Aprog^2)$ is obtained by moving only awaits smaller than $\statement_{w}^{2}$ w.r.t. $\wao$. Thus, in $\Aprog^{2'}$ the await $\statement_{w}^{1}$ is in the same position as in $\Aprog'$. Then, $\Aprog^{1'} = \textsc{MaxRel}(\Aprog^1)$ is different than $\Aprog^{2'}$. For any two programs $\Aprog^{1''}$ and $\Aprog^{2''}$ s.t. $\Aprog^{1''}$ (resp., $\Aprog^{2''}$) is outputted by \emph{\textsc{Asy\-Syn}($\Aprog^{1'},\statement_{w}^{1}$)} (resp., \emph{\textsc{Asy\-Syn}($\Aprog^{2'},\statement_{w}^{2}$)}), we have that the two programs are distinct since in $\Aprog^{2''}$ the await $\statement_{w}^{1}$ is in the same position as in $\Aprog'$. Thus, we get that \emph{\textsc{Asy\-Syn}} outputs every element of $\mathcal{A}$ only once. 
\end{proof}

\section{Formalization and Proofs of Section~\ref{sec:6}} \label{sec:Appendix6}

\begin{wrapfigure}{r}{0.21\textwidth}
  \vspace{-1.2cm}
  \centering
  \lstset{basicstyle=\ttfamily\scriptsize,numberblanklines=false,mathescape=true,morekeywords={async,method,await}}
  \begin{minipage}[l]{0.97\linewidth}
  \begin{lstlisting} 
async method Main {  
  r1 = call m;
  x = 1;
  await r1;
} 
async method m {
  r2 = call m1;
  r3 = call m1;
  await r2;
  r4 = x;
  if r4 == 1
    y = 2;
  await r3;
}
async method m1 {
  await *;
  r5 = y;
  return;
}
\end{lstlisting}
\end{minipage}
\vspace{-0.5cm}
  \caption{}
  \label{fig:dataRacesExample10}
  \vspace{-1cm}
\end{wrapfigure}
The following lemma proves that for any unsound asynchronization, any trace with a data race contains at least one data race that involves two actions that are synchronously reachable. For instance, the program in Fig.~\ref{fig:dataRacesExample10} has two data races, one between $x = 1$ and $r4 = x$ and the other between $y=2$ and $r5=y$. However, the statement $y = 2$ is not reachable in the corresponding synchronous program. It is reachable in this asynchronization because of the data race between $x = 1$ and $r4 = x$, which are both reachable in the synchronous program. Eliminating the latter data race by moving the statement $\awaitact\ r1$ before $x =1$, makes $y = 2$ unreachable and the data race between $y=2$ and $r5=y$ is also eliminated.

\vspace{-1mm}
\begin{lemma}\label{lem:races}
An asynchronization $\Aprog[\Alib]$ is sound iff it does not admit data races between actions that are synchronously reachable.
\vspace{-1mm}
\end{lemma}

\begin{proof}[Proof of Lemma~\ref{lem:races}]
  Assume by contradiction that $\Aprog[\Alib]$ is sound and it admits a data race $(\action_1,\action_2)$ in a trace $\tau\in \tracesconf(\Aprog[\Alib])$ where one of the actions, say $\action_1$, is not synchronously reachable. We assume w.l.o.g that the data race $(\action_1,\action_2)$ is the first that occurs in $\tau$ with at least one synchronously unreachable action. Then, there must exist a read access $\action_r$ that enabled $\action_1$, and therefore, $\action_r$ reads a value that was not read in any synchronous execution. Thus, the read value must the result of another data race that occurs earlier in the trace $\tau$, which is a contradiction. 
\end{proof}

  \begin{wrapfigure}{r}{0.21\textwidth}
    \vspace{-0.8cm}
    \centering
    \lstset{basicstyle=\ttfamily\scriptsize,numberblanklines=false,mathescape=true,morekeywords={async,method,await}}
    \begin{minipage}[l]{0.97\linewidth}
    \begin{lstlisting} 
async method Main {
 r1 = call m;
 r2 = x;
 x = r2 + 1;
 await r1;
}  
async method m {
 await *; 
 x = 2;
 return;
}
  \end{lstlisting}
  \end{minipage}
  \vspace{-0.8cm}
    \caption{}
    \label{fig:dataRacesExample01}
    \vspace{-0.8cm}
  \end{wrapfigure}
In Fig.~\ref{fig:dataRacesExample01}, we explain how the repairing data races based on the partial order relating data races  allows to avoid superfluous repair steps. For instance, in Fig. \ref{fig:dataRacesExample01}, the first data race to repair involves the read of \texttt{x} from \texttt{Main} and the write to \texttt{x} in \texttt{m}, because these statements are the first to execute in the original sequential program among the other statements involved in data races. Repairing this data race consists in moving \texttt{await r1} before the read of \texttt{x} from \texttt{Main}, which implies that \texttt{m} completes before the read of \texttt{x}. This repair is defined from a notion of \emph{root cause} of a data race, that in this case, contains the call to \texttt{m} and the read of \texttt{x} from \texttt{Main}. Interestingly, this repair step removes the write-write data race between the write to \texttt{x} in \texttt{Main} and the write to \texttt{x} in \texttt{m} as well. If we would have repaired these data races in the opposite order, we would have moved \texttt{await t1} first before the write to \texttt{x}, and then, before the read of \texttt{x}.

\begin{wrapfigure}{r}{0.15\textwidth}
  \vspace{-0.1cm}
  \lstset{basicstyle=\ttfamily\scriptsize,numberblanklines=false,mathescape=true,morekeywords={async,method,await}}
  \centering
  \begin{lstlisting} 
method Main {         
  while $*$
    if $*$
      r1 = x;
    r2 = y; 
}  
  \end{lstlisting}
  \vspace{-0.7cm}
  \caption{}
  \label{fig:nondeterministicprogram}
  \vspace{-0.7cm}
\end{wrapfigure}
In Fig.~\ref{fig:nondeterministicprogram}, we give a non-deterministic program where two statements of the program can be executed in different orders in different executions. In particular, the statements \texttt{r1 = x} and \texttt{r2 = y} of the program can be executed in different orders depending on the number of loop iterations and whether the if branch is entered during the first loop iteration. 

\begin{wrapfigure}{r}{0.21\textwidth}
  \vspace{-0.9cm}
  \centering
  \lstset{basicstyle=\ttfamily\scriptsize,numberblanklines=false,mathescape=true,morekeywords={async,method,await}}
  \begin{minipage}[l]{0.99\linewidth}
  \begin{lstlisting} 
async method Main {  
  r1 = call m;
  if $*$
    r2 = x;
    x = r2 + 1;     
  else 
    r3 = x;
  await r1;
} 
async method m {
  await $*$
  retVal = x;
  x = input;
  return;
}
\end{lstlisting}
\end{minipage}
\vspace{-0.5cm}
  \caption{}
  \label{fig:dataRacesExample2}
  \vspace{-0.7cm}
\end{wrapfigure}
For the program in Fig.~\ref{fig:dataRacesExample2}, we have the following order between data races: $(\texttt{x = input}, \texttt{r2 = x})\prec_\so$$(\texttt{retVal = x},$ $\texttt{x = r2 + 1})$ because \texttt{r2 = x} is executed before the write \texttt{x = r2 + 1} in the original synchronous program (for simplicity we use statements instead of actions). However, the data races $(\texttt{x = input},\texttt{r2 = x})$ and $(\texttt{x = input},\texttt{r3 = x})$ are incomparable.


The following lemma identifies a sufficient transformation for repairing a data race $(\action_1, \action_2)$: moving the await $\statement_w$ generating the action $\action_w$ just before the statement $\statement$ generating $\action$. This is sufficient because it ensures that every statement that follows $\lca_\co(\action_1,\action_2)$\footnote{We abuse the terminology and make no distinction between statements and actions.} in call order will be executed before $\action$ and before any statement which succeeds $\action$ in call order, including $\action_2$. Note that moving the await $\action_w$ anywhere after $\action$ will not affect the concurrency between $\action_1$ and $\action_2$. 

\begin{lemma} \label{lemma:anc2}
  Let $(\action_1, \action_2)$ be a data race in a trace $\tau$ of an asynchronization $\Aprog$, and $\action_c = (\aniden, \callact(\anidenp)) = \lca_\co(\action_1,\action_2)$. Then, $\tau$ contains a unique action $\action_w = (\aniden, \awaitact(\anidenp))$ and a unique action $\action$ such that:
  \begin{itemize}[noitemsep,topsep=0pt]
    \item $(\action,\action_w)\in\mo$, and $\action$ is the latest action in the method order $\mo$ such that $(\action_c,\action)\in \mo$ and $(\action,\action_2)\in \co^*$ ($\co^*$ denotes the reflexive closure of $\co$).
  \end{itemize}
  \vspace{-1mm}
\end{lemma}

\begin{proof}[Proof of Lemma \ref{lemma:anc2}]
  Let $\rho$ be the execution of the trace $\tau$. By definition, $\rho$ ends with a configuration where the call stack and the set of pending tasks are empty. Therefore, $\rho$ contains an action $\action_w = (\_,\aniden, \awaitact(\anidenp))$ matching $a_c$ which is unique by the definition of the semantics.
  Since $(\action_c, \action_1) \in \co$ and $(\action_c, \action_2) \in \co$ then either 
  $\action_c$ and $\action_2$ occur in the same method, or there exists a call action $\action'$ in the same task as $\action_c$ such that $(\action',\action_2)\in \co$.
  Then, we define $\action=\action_2$ in the first case, and $\action$ as the latest action in the same task as $\action_c$ such that $(\action,\action_2)\in \co$ in the second case.
  We have that $(\action,\action_w)\in\mo$ because otherwise, $(\action_w,\action)\in\mo$ and $(\action,\action_2)\in \co^*$ implies that $(\action_1,\action_2)\in\hbo{}$ (because $(\action_1,\action_w)\in\hbo{}$, and $\mo$ and $\co$ are included in $\hbo{}$), and this contradicts $\action_1$ and $\action_2$ being concurrent.
\end{proof}

When the control-flow graph of the method contains branches, the construction of $\repairdatarace(\Aprog, \statement_c,\statement)$ involves (1) replacing all \plog{await} statements matching $\statement_c$ that are reachable in the CFG from $\statement$ with a single \plog{await} statement placed just before $\statement$, and (2) adding additional \plog{await} statements in branches that ``conflict'' with the branch containing $\statement$. This is to ensure the syntactic constraints described in Section \ref{sec:1}. These additional \plog{await} statements are at maximal distance from the corresponding call statement because of the maximality requirement.

\begin{wrapfigure}{r}{0.49\textwidth}
  \vspace{-0.2cm}
  \lstset{basicstyle=\ttfamily\scriptsize,numbers=none,
          stepnumber=1,numberblanklines=false,mathescape=true,morekeywords={async,method,await}}
  \begin{minipage}[t]{0.48\linewidth}
  \begin{lstlisting} 
async method Main {  
  r1 = call m;
  if $*$

    r2 = x;
  else 
    r3 = y;

  await r1;
} 
async method m {
  await $*$
  retVal = x;
  x = input;
  return;
}
\end{lstlisting}
  \end{minipage}
  \hfill
  \begin{minipage}[t]{0.48\linewidth}
    \begin{lstlisting} 
async method Main {  
  r1 = call m;
  if $*$
    await r1;
    r2 = x;
  else 
    r3 = y;
    await r1;

} 
async method m {
  await $*$
  retVal = x;
  x = input;
  return;
}
  \end{lstlisting}
  \end{minipage}
    \vspace{-0.9cm}
    \caption{Examples of asynchronizations.}
    \label{fig:dataRacesExamples}
    \vspace{-0.8cm}
\end{wrapfigure}
 For instance, to repair the data race between {\tt r2 = x} and {\tt x = input} in the program on the left of Fig.~\ref{fig:dataRacesExamples}, the statement \plog{await} \texttt{r1} must be moved before \texttt{r2 = x} in the \plog{if} branch, which implies that another await must be added on the \plog{else} branch. The result is given on the right of Fig.~\ref{fig:dataRacesExamples}.

The following lemma shows that repairing a minimal data race cannot introduce smaller data races (w.r.t. $\prec_\so$), which ensures some form of monotonicity when repairing minimal data races iteratively.

\vspace{-1mm}
\begin{lemma}\label{lemma:repairMono}
Let $\Aprog$ be an asynchronization, $(\action_1,\action_2)$ a data race in $\Aprog$ that is minimal w.r.t. $\prec_\so$, and $(\statement_c,\statement)$ the root cause of $(\action_1,\action_2)$. Then, $\repairdatarace(\Aprog, \statement_c,\statement)$ does not admit a  data race that is smaller than $(\action_1,\action_2)$ w.r.t. $\prec_\so$.
\vspace{-1mm}
\end{lemma}


\begin{wrapfigure}{r}{0.45\textwidth}
  \centering
  \vspace{-0.1cm}
      \includegraphics[width=0.45\textwidth]{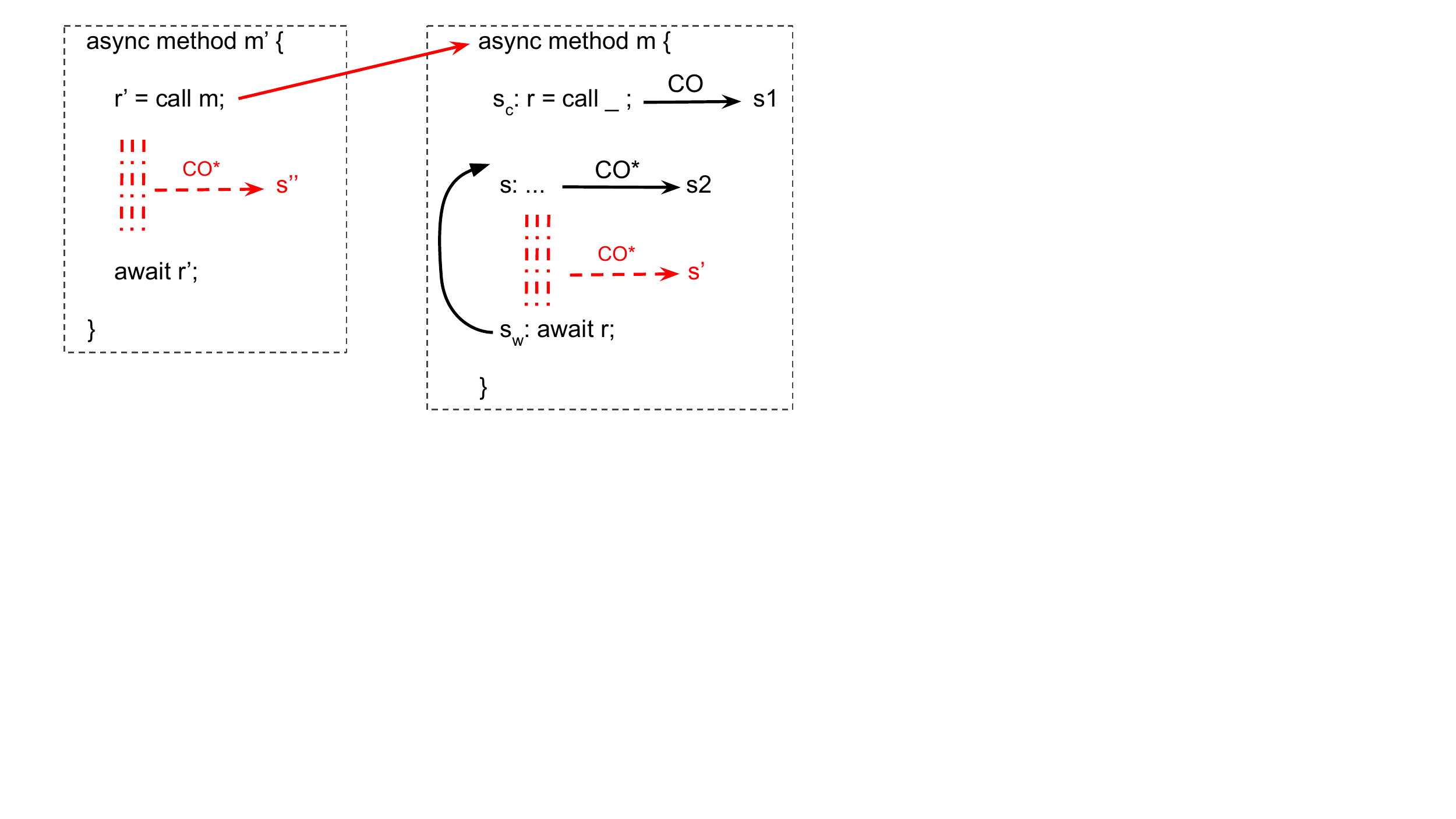}
  \vspace{-0.6cm}
      \caption{An excerpt of an asynchronous program.}
      \label{fig:proof}
      \vspace{-0.7cm}
\end{wrapfigure}
\textsc{Proof of Lemma~\ref{lemma:repairMono}.}
The only modification in the program $\Aprog'=\repairdatarace(\Aprog, \statement_c,\statement)$ compared to $\Aprog$ is the movement of the await $\statement_w$ matching the call $\statement_c$ to be before the statement $\statement$ in a method $\ameth$. The concurrency added in $\Aprog'$ that was not possible in $\Aprog$ is between actions $(\action',\action'')$ generated by statements $\statement'$ and $\statement''$, respectively, as shown in Fig.~\ref{fig:proof}. W.l.o.g., we assume that $(\action',\action'') \in \so$.
  The statements $\statement_1$ and $\statement_2$ are those generating $\action_1$ and $\action_2$, respectively.
  The statement $\statement'$ is related by $\co^*$ to some statement in $\ameth$ that follows $\statement$, and $\statement''$ is related by $\co^*$ to some statement that follows the call to $\ameth$ in the caller of $\ameth$. Note that $\statement'$ is ordered by $\prec$ after $\statement_2$. 
  Since $(\action_1,\action_2) \in \so$ and $(\action',\action'') \in \so$ then  $\statement_2\prec \statement''$ and $\statement_1\prec \statement'$. Thus, any new data race $(\action',\action'')$ in $\Aprog'$ that was not reachable in $\Aprog$ is bigger than $(\action_1,\action_2)$. 
\hfill $\Box$

\begin{theorem}\label{theorem:anc3}
  Given an asynchronization $\Aprog\in \Asyof{\aprog, \alib, \Alib}$, \textsc{MaxRel}($\Aprog$) returns the optimal asynchronization of $\aprog$ relative to $\Aprog$.
\end{theorem}

\begin{proof}[Proof of Theorem~\ref{theorem:anc3}]
  Since the recursive calls $\textsc{RCMinDR}$ find all data races between synchronously reachable actions then the output $\Aprog'=\textsc{MaxRel}(\Aprog)$ is sound and therefore it is equivalent to $\aprog$ (Lemma \ref{lemma:soundDR} and Lemma \ref{lem:races}). Now we need to show that any successor $\Aprog^1$ of $\Aprog'$ that is also smaller than $\Aprog$ (w.r.t. $\leq$) admits data races. Let $\statement_w$ be the biggest await statement w.r.t. $\wao$ whose position in $\Aprog^1$ is changed with respect to its position in $\Aprog'$ (moved down). Since $\Aprog^1\leq \Aprog$, then $\statement_w$ was also moved up by the procedure $\textsc{MaxRel}$ with respect to its position in $\Aprog$ to fix some data race $(\action_1,\action_2)$. Let $\ameth$ be the method $\ameth$ that contains $\statement_w$ and $\statement_c$ be the matching call. We will now show that $(\action_1,\action_2)$ forms a data race in $\Aprog^1$ as well.  $\Aprog^1$ has an execution $\rho$ that reaches both $\action_1$ and $\action_2$ (since $\executionsconf(\Aprog^1)$ includes the synchronous execution where all $\awaitact\ *$ are interpreted as skip which reaches $\action_1$ and $\action_2$). Since every other await $\statement'_w$ in $\Aprog^1$ that occurs in a method $\ameth'$ (in)directly called by $\ameth$ (including the method associated with the call $\statement_c$) is in the same position as in $\Aprog'$, then the two actions $\action_1$ and $\action_2$ are not related by $\hbo$ and are concurrent. Thus, $(\action_1,\action_2)$ forms a data race in $\Aprog^1$, which concludes the proof.
\end{proof}

The fact that data races are enumerated in the order defined by $\prec_\so$ guarantees a bound on the number of times an await matching the same call is moved during the execution of  \textsc{MaxRel}($\Aprog$). In general, this bound is the number of statements covered by all the awaits matching the call in the input program $\Aprog$. Actually, this is a rather coarse bound. A more refined analysis has to take into account the number of branches in the CFGs. For programs without conditionals or loops, every await is moved at most once during the execution of  \textsc{MaxRel}($\Aprog$). In the presence of branches, a call to an asynchronous method may match multiple \plog{await} statements (one for each CFG path starting from the call), and the data races that these \plog{await} statements may create may be incomparable w.r.t. $\prec_\so$. 
Therefore, for a call statement $\statement_c$, let $|\statement_c|$ be the sum of $|\cover(\statement_w)|$ for every await $\statement_w$ matching $\statement_c$ in $\Aprog$.

\begin{lemma}\label{lem:await2}
  For any asynchronization $\Aprog\in \Asyof{\aprog, \alib, \Alib}$ and call statement $\statement_c$ in $\Aprog$, the while loop in \textsc{MaxRel}($\Aprog$) does at most $|\statement_c|$ iterations that result in moving an await matching $\statement_c$.
\end{lemma}

\begin{proof}[Proof of Lemma~\ref{lem:await2}]
  We consider first the case without conditionals or loops, and we show by contradiction that every \plog{await} statement $\statement_w$ is moved at most once during the execution of  \textsc{MaxRel}($\Aprog$), i.e., there exists at most one iteration of the while loop which changes the position of $\statement_w$. Suppose that the contrary holds for an \plog{await} $\statement_w$. Let $(\action_1,\action_2)$, and $(\action_3,\action_4)$ be the data races repaired by the first and second moves of $\statement_w$, respectively. 
  By Lemma \ref{lemma:anc2}, there exist two actions $\action$ and $\action'$ such that 
  \begin{align*}
  (\action_c,\action)\in \mo,\ (\action,\action_2)\in \co^*,(\action,\action_w)\in\mo\mbox{ and }
  (\action_c,\action')\in \mo,\ (\action',\action_4)\in \co^*, (\action',\action_w)\in\mo
  \end{align*}
  where  $\action_w = (\_,\aniden, \awaitact(\anidenp))$ and $\action_c = (\_,\aniden, \callact(\anidenp))$ are the asynchronous call action and the matching  await action. 
  Let $\statement_2$ and $\statement_4$ be the statements generating the two actions $\action_2$ and $\action_4$, respectively.
  Then, we have either  $\statement_2\prec \statement_4$ or $\statement_2 = \statement_4$, and both cases imply that $(\action,\action')\in \mo^*$. Thus, moving the await statement generating $\action_w$ before the statement generating $\action$ implies that it is also placed before the statement generating $\action'$ (that occurs after $\action$ in the same method). Thus, the first move of the await $\statement_w$ repaired both data races, which is contradiction. 
  
  In the presence of conditionals or loops, moving an await up in one branch may correspond to adding multiple awaits in the other conflicting branches. Also, one call in the program may correspond to multiple awaits on different branches. However, every repair of a data race consists in moving one await closer to the matching call $\statement_c$ and before one more statement covered by some await matching $\statement_c$ in the input $\Aprog$.
  \end{proof}
\section{Computing Root Causes of Minimal Data Races} \label{sec:recApp}
We present a reduction from the problem of computing root causes of minimal data races to reachability (assertion checking) in sequential programs. 
  This reduction builds on a program instrumentation for checking if there exists a minimal data race that involves two given statements $(\statement_1,\statement_2)$ that are reachable in an execution of the original synchronous program, whose correctness relies on the assumption that another pair of statements cannot produce a smaller data race. This instrumentation is used in an iterative process where pairs of statements are enumerated according to the colexicographic order induced by $\prec$. This specific enumeration ensures that the assumption made for the correctness of the instrumentation is satisfied.
  
\begin{figure}[t]
  \lstset{basicstyle=\ttfamily\scriptsize,numbers=left,
      stepnumber=1,numberblanklines=false,mathescape=true,escapechar=@,morekeywords={assert}}
  \begin{minipage}[t]{0.51\textwidth}
  \begin{lstlisting}[xleftmargin=4mm]
Add before $\statement_1$:
if ( lastTaskDelayed == $\bot$ && * ) @\label{ln:s1Start}@
  lastTaskDelayed := myTaskId();
  DescendantDidAwait := thisHasDoneAwait;
  return @\label{ln:s1End}@

Add before $\statement_2$:
  if ( task_$\statement_c$ == myTaskId() )
    $\statement$ := $\statement_2$; @\label{ln:S2S}@
  assert (lastTaskDelayed == $\bot$ || !DescendantDidAwait); @\label{ln:assert}@
\end{lstlisting}
  \end{minipage}
  \hfill
  \begin{minipage}[t]{0.465\textwidth}
  \begin{lstlisting}[firstnumber=13]
Replace every statement ``await r'' with:
  if( r == lastTaskDelayed  ) then @\label{ln:awaitStart}@
    if ( !DescendantDidAwait )
      DescendantDidAwait := thisHasDoneAwait; 
    lastTaskDelayed := myTaskId();
    return @\label{ln:awaitReturn}@
  else  
    thisHasDoneAwait := true @\label{ln:awaitEnd}@

Add before every statement ``r := call m'':
  if ( task_$\statement_c$ == myTaskId() ) then @\label{ln:callS}@
    s := this statement;

Add after every statement ``r := call m'':
  if ( r == lastTaskDelayed ) @\label{ln:callSC}@
    $\statement_c$ := this statement;
    task_$\statement_c$ := myTaskId();
\end{lstlisting}
  \end{minipage}
  \vspace{-0.4cm}
  \caption{A program instrumentation for computing the root cause of a minimal data race between the statements $\statement_1$ and $\statement_2$ (if any). All variables except for \texttt{thisHasDoneAwait} are program (global) variables. \texttt{thisHasDoneAwait} is a local variable. The value $\bot$ represents an initial value of a variable. The variables $\statement_c$ and $\statement$ store the (program counters of the) statements representing the root cause. The method \texttt{myTaskId} returns the id of the current task. }
  \label{fig:instrumentation}
\end{figure}

  Given an asynchronization $\Aprog$, the instrumentation described in Fig.~\ref{fig:instrumentation} represents a synchronous program where all \plog{await} statements are replaced with synchronous code (lines~\ref{ln:awaitStart}--\ref{ln:awaitEnd}). This instrumentation simulates asynchronous executions of $\Aprog$ where methods may be only partially executed, modeling \plog{await} interruptions. It reaches an error state (see the assert at line~\ref{ln:assert}) when an action generated by $\statement_1$ is concurrent with an action generated by $\statement_2$, which represents a data race, provided that $\statement_1$ and $\statement_2$ access a common program variable (these statements are assumed to be given as input). Also, the values of $\statement_c$ and $\statement$ when reaching the assertion violation represent the root-cause of this data race.
  
  The instrumentation simulates an execution of $\Aprog$ to search for a data race as follows (we discuss the identification of the root-cause afterwards): 
  \begin{itemize}[noitemsep,topsep=0pt]
    \item It executes under the synchronous semantics until an instance of $\statement_1$ is non-deterministically chosen as a candidate for the first action in the data race ($\statement_1$ can execute multiple times if it is included in a loop for instance). The current invocation is interrupted when it is about to execute this instance of $\statement_1$ and its task id $t_0$ is stored into \texttt{lastTaskDelayed} (see lines~\ref{ln:s1Start}--\ref{ln:s1End}).
    \item Every invocation that transitively called $t_0$ is interrupted when an \plog{await} for an invocation in this call chain (whose task id is stored into \texttt{lastTaskDelayed}) would have been executed in the asynchronization $\Aprog$ (see line~\ref{ln:awaitReturn}). 
    \item Every other method invocation is executed until completion as in the synchronous semantics.
    \item When reaching $\statement_2$, if $\statement_1$ has already been executed (\texttt{lastTaskDelayed} is not $\bot$) and at least one invocation has only partially been executed, which is recorded in the boolean flag \texttt{DescendantDidAwait} and which means that $\statement_1$ is concurrent with $\statement_2$, then the instrumentation stops with an assertion violation.
  \end{itemize}
  A subtle point is that the instrumentation may execute code that follows an \plog{await} $r$ even if the task $r$ has been executed only partially, which would not happen in an execution of the original $\Aprog$. Here, we rely on the assumption that there exist no data race between that code and the rest of the task $r$. Such data races would necessarily involve two statements which are before $\statement_2$ w.r.t. $\prec$. Therefore, the instrumentation is correct only if it is applied by enumerating pairs of statements $(\statement_1,\statement_2)$ w.r.t. the colexicographic order induced by $\prec$.
  
  Next, we describe the computation of the root-cause, i.e., the updates on the variables $\statement_c$ and $\statement$.  
  By definition, the statement $\statement_c$ in the root-cause should be a call that makes an invocation that is in the call stack when $\statement_1$ is reached. This can be checked using the variable \texttt{lastTaskDelayed} that stores the id of the last such invocation popped from the call stack (see the test at line~\ref{ln:callSC}). The statement $\statement$ in the root-cause can be any call statement that has been executed in the same task as $\statement_c$ (see the test at line~\ref{ln:callS}), or $\statement_2$ itself (see line~\ref{ln:S2S}).

  Let $\sem{\Aprog,\statement_1,\statement_2}$ denote the instrumentation in Fig.~\ref{fig:instrumentation}. We say that the values of $\statement_c$ and $\statement$ when reaching the assertion violation are the root cause computed by this instrumentation. 
  The following theorem states its correctness.
  
  \vspace{-1mm}
  \begin{theorem}\label{thm:RootCause}
  If $\sem{\Aprog,\statement_1,\statement_2}$ reaches an assertion violation, then it computes the root cause of a minimal data race, or there exists $(\statement_3,\statement_4)$ such that $\sem{\Aprog,\statement_3,\statement_4}$ reaches an assertion violation and $(\statement_3,\statement_4)$ is before $(\statement_1,\statement_2)$ in colexicographic order w.r.t. $\prec$.
  \vspace{-1mm}
  \end{theorem}
  
  Based on Theorem~\ref{thm:RootCause}, we define an implementation of the procedure $\textsc{RCMinDR}(\Aprog)$ used in computing maximal asynchronizations (Algorithm~\ref{algo2}) as follows:
  \setlength{\itemsep}{1mm}
  \begin{itemize}
    \item For all pairs of read or write statements $(\statement_1,\statement_2)$ in colexicographic order w.r.t. $\prec$ that are reachable in an execution of the original synchronous program $\aprog$.
    \vspace{1mm}
    \begin{itemize}
      \item If $\sem{\Aprog,\statement_1,\statement_2}$ reaches an assertion violation, then 
      \begin{itemize}
        \item return the root cause computed by $\sem{\Aprog,\statement_1,\statement_2}$ 
      \end{itemize}
    \end{itemize}
    \item return $\bot$
  \end{itemize}
  
  Checking whether read or write statements are reachable can be determined using a linear number of reachability queries in the synchronous program $\aprog$. Also, the order $\prec$ between read or write statements can be computed using a quadratic number of reachability queries in the synchronous program $\aprog$. Therefore, $\statement\prec\statement'$ iff an instrumentation of $\aprog$ that sets a flag when executing $\statement$ and asserts that this flag is not set when executing $\statement'$ reaches an assertion violation. The following theorem states the correctness of the procedure above.
  
  \begin{theorem}
  $\textsc{RCMinDR}(\Aprog)$ returns the root cause of a minimal data race of $\Aprog$ w.r.t. $\prec_\so$, or $\bot$ if $\Aprog'$ is data race free.
  \end{theorem}

\section{Formalization and Proofs of Section~\ref{sec:asymComplexity}} \label{appendixF}

\begin{theorem}\label{th:compl3}
Checking whether there exists a sound asynchronization different from the strong asynchronization is PSPACE-complete.
\end{theorem}

\begin{proof}[Proof of Theorem \ref{th:compl3}]
  (1) define a new method $\ameth$ that writes to a new program variable $x$, and insert a call to $\ameth$ followed by a write to $x$ at location $\ell$, and 
 (2) insert a write to $x$ after every call statement that calls a method in $\{\ameth'\}^*$, where $\ameth'$ is the method containing $\ell$.
  Let $\ameth_a$ be an asynchronous version of $\ameth$ obtained by inserting an \plog{await} $*$ at the beginning. Then, $\ell$ is reachable in $\aprog$ iff the only sound asynchronization of $\aprog'$ w.r.t. $\{\ameth_a\}$ is the strong asynchronization.
\end{proof}

\vspace{-5pt}
\section{Formalization and Proofs of Section \ref{sec:analysis}}\label{sec:appanalysis}

The \textsc{MaxRel}$^\#$ procedure repairs data races in an order which is $\prec_\so$ 
with some exceptions that do not affect optimality, i.e., the number of times an await matching the same call can be moved. 
For instance, if a method $\ameth$ calls two other methods $\ameth_1$ and $\ameth_2$ in this order, the procedure above may 
handle $\ameth_2$ before $\ameth_1$, i.e., repair data races between actions that originate from $\ameth_2$ before data 
races that originate from $\ameth_1$, although the former are bigger than the latter in $\prec_\so$. 
This does not affect optimality because those repairs are ``independent'', i.e., any repair in $\ameth_2$ cannot 
influence a repair in $\ameth_1$, and vice-versa. The crucial point is that this procedure repairs data races between 
actions that originate from a method $\ameth$ before data races that involve actions in methods preceding $\ameth$ in 
the call graph, which are bigger in $\prec_\so$ than the former.

Note that \textsc{MaxRel}$^\#$ procedure which is based on the bottom-up inter-procedural data-flow analysis compromises precision to reduce the complexity of the problem from undecidable in general or PSPACE-complete with finite data to polynomial time. However, because of this imprecision, certain await statements may be moved closer to the matching call unnecessarily. For instance, in Fig.~\ref{fig:dataRacesExample10}, the precise algorithm (using the procedure \textsc{MaxRel} in Algorithm~\ref{algo2}) will only repair the data race on $x$ because doing so, the potential data race on $y$ will become unreachable. On the other hand, the polynomial-time algorithm (using the \textsc{MaxRel}$^\#$ procedure) will also repair the data race on $y$, moving another await closer to the matching call, since it cannot reason about data (one statement of this data race is only reachable if the variable $r4$ is $1$).  

\begin{figure}[h!]
\lstset{basicstyle=\ttfamily\scriptsize,numbers=left,
        stepnumber=1,numberblanklines=false,mathescape=true}
\hspace{3mm}
\begin{minipage}[t]{0.23\textwidth}
\begin{lstlisting}[numbersep=2pt]
void Main() {
 F();

 x = 2;
 
} 

void F() {
 IO();
  
 x = 1;

}
\end{lstlisting}
\end{minipage}
\hfill
\begin{minipage}[t]{0.30\textwidth}
    \begin{lstlisting}[numbersep=2pt]
async Task MainAsync() {
 Task t1 = F();

 x = 2; $\label{ln:multix2}$
 await t1;
}  

async Task F() {
 Task t2 = IOAsync();

 x = 1; $\label{ln:multix1}$
 await t2;
}
\end{lstlisting}
\end{minipage}
\hfill
\begin{minipage}[t]{0.38\textwidth}
\begin{lstlisting}[xleftmargin=2.5mm,numbersep=2pt]
void Main() {
 Thread thr1 = new Thread(F);
 thr1.Start(); 
 x = 2;
 thr1.Join(); $\label{ln:join1}$
}  

void F() {
 Thread thr2 = new Thread(IO);
 thr2.Start(); 
 x = 1;
 thr2.Join();
}
\end{lstlisting}
\end{minipage}
\vspace{-0.4cm}
\caption{A synchronous C\# program, an asynchronization, and a multi-threaded refactoring.}
\label{fig:example02}
\end{figure}

\section{Multi-threaded Refactorings}\label{sec:startJoin}

We discuss an extension of our framework to \emph{multi-threaded refactorings} that rewrite a sequential program into a multi-threaded program where every method invocation is executed on a different thread. A caller can wait for a callee to complete using a \texttt{join} primitive. A \texttt{start} primitive for spawning a new thread is the counterpart of an asynchronous call while \texttt{join} is the counterpart of \texttt{await}. For instance, 
Fig.~\ref{fig:example02} lists a sequential program, a possible asynchonization, and a multi-thread refactoring (both refactorings place the awaits/joins as far away as possible from the calls).

An important difference between start/join and async/await is the happens-before order relation. For instance, the asynchronization on the center of Fig.~\ref{fig:example02} assigns 1 to \texttt{x} (line \ref{ln:multix1}) before it assigns 2 to \texttt{x} (line \ref{ln:multix2}), as in the original sequential program. However, the multi-thread program on the right of Fig.~\ref{fig:example02} may execute these two assignments in any order, and admits a behavior that is not possible in the sequential program (assigning 2 before assigning 1). 
Repairing this data-race consists in moving the \texttt{join} at line \ref{ln:join1} to occur before assigning 2 to \texttt{x} at line \ref{ln:multix2}. In general, the happens-before order is weaker compared to an analogous asynchronization, where awaits are placed as the joins, which implies that any multi-threaded refactoring can be rewritten to an asynchronization. The vice-versa may not be possible as shown in this example.

Despite this difference, it can still be proved that there exists a unique multi-threaded refactoring that is sound, i.e., does not admit data races, and maximal, i.e., maximizes the distance between start and join, a result similar to Lemma~\ref{lemma:optimalAsync}. Assuming by contradiction the existence of two incomparable maximal and sound refactorings, one can show that moving a join in one refactoring further away from the matching call as in the other refactoring does not introduce data races (contradicting optimality).
To compute maximal and sound multi-threaded refactorings, one can apply the same 
iterative process of repairing data-races (the happens-before reflects multi-threading instead of async/await), prioritizing data races involving statements that would execute first in the sequential program. The repairing of a data-race is similar and consists in moving a join up.

In contrast to async/await, moving a join up does not introduce new data races (since no new parallelism is introduced). This implies that all the predecessors of a sound multi-threaded refactoring are also sound, i.e., the set of sound multi-threaded refactorings is downward closed. 

\end{document}